%% file: arxiv_updated.tex
\documentclass[pra,aps,reprint,a4paper,superscriptaddress,floatfix,nofootinbib]{revtex4-2}
\usepackage{times} 
\usepackage{graphicx}
\usepackage{subfig}
\usepackage{epsfig}
\usepackage{amsfonts}
\usepackage{amsmath}
\usepackage{amssymb}
\usepackage{dsfont}
\usepackage{amsthm}
\usepackage{color}
\usepackage{comment}
\usepackage{braket}
\usepackage{physics}
\usepackage{multirow}
\usepackage{todonotes}
\usepackage{float}
\usepackage[colorlinks=true,linkcolor=blue,citecolor=blue,urlcolor=blue]{hyperref}
\usepackage{footnotebackref}

\makeatletter

\let\FC@parse\NAT@parse
\def\NAT@parse#1{
  \@ifundefined{fcite@#1}{
    \global\@namedef{fcite@#1}{}
    \Hy@raisedlink{\hypertarget{fcite.#1}{}}
  }{}%
  \FC@parse{#1}
}

\AtBeginDocument{%
  \let\FC@lbibitem\@lbibitem
  \def\@lbibitem[#1]#2{\let\NAT@parse\FC@parse\FC@lbibitem[#1]{#2}}%
}

\def\NAT@anchor#1#2{%
 \hyper@natanchorstart{#1\@extra@b@citeb}%
  \def\@tempa{#2}\@ifx{\@tempa\@empty}{}{%
    \@ifundefined{fcite@#1}%
      {\@biblabel{#2}}%
      {\begingroup\def\@linkcolor{black}\hyperlink{fcite.#1}{\@biblabel{#2}}\endgroup}%
  }%
 \hyper@natanchorend
}
\makeatother

\usepackage{capt-of}   
\usepackage[margin=1in]{geometry}

\newtheorem{theorem}{Theorem}

\newtheorem{corollary}{Corollary}

\newtheorem{remark}{Remark}

\usepackage[ruled,linesnumbered]{algorithm2e}
\usepackage{mathdots}
\usepackage{MnSymbol}

\usepackage{tikz}
\usetikzlibrary{calc,fit}
\DeclareRobustCommand{\stirling}{\genfrac\{\}{0pt}{}}

\SetCommentSty{mycommfont}

\makeatletter
\let\FM@foot@orig\frontmatter@footnote
\let\@notesep\relax                 
\def\@author@join@#1#2#3{
  \@if@empty{#2}%
    {\def\@author{{#1}{#3}}}%
    {\def\@author{{#1}{#2\@notesep#3}}}%
}
\def\frontmatter@footnote#1{\FM@split#1\@notesep\@nil}
\def\FM@split#1\@notesep#2\@nil{%
  \FM@foot@orig{#1}%
  \@if@empty{#2}{}{\FM@split#2\@nil}%
}
\makeatother

\begin{document}


 \title{Bell inequalities from outcome equalities}


\author{Ricardo Faleiro}
\thanks{These authors contributed equally to this work.}
\affiliation{Instituto de Telecomunicações, Av. Rovisco Pais, 1049-001, Lisboa, Portugal}

\author{Flavien Hirsch}
\thanks{These authors contributed equally to this work.}
  \email{flavien.uni@gmail.com}
\affiliation{Instituto de Telecomunicações, Av. Rovisco Pais, 1049-001, Lisboa, Portugal}

\author{Emmanuel Zambrini Cruzeiro}
\thanks{These authors contributed equally to this work.}
\affiliation{Instituto de Telecomunicações, Av. Rovisco Pais, 1049-001, Lisboa, Portugal}
\affiliation{Departamento de Engenharia Electrotécnica e de Computadores,
Instituto Superior Técnico, Av. Rovisco Pais, 1049-001, Lisbon, Portugal}

\author{Nicolas Gisin}
\email{nicolas.gisin@unige.ch}
\affiliation{Group of Applied Physics, University of Geneva, 1211 Geneva 4, Switzerland}
\affiliation{Constructor University, Bremen, Germany}


\begin{abstract}

In this work, we study a particular class of Bell inequalities involving only direct equality-comparisons of outcomes. In the bipartite case, the scenario can be interpreted as a natural generalization of full-correlator inequalities (XOR games) beyond binary outputs.  We define the sub-polytope of the local polytope corresponding to this scenario and solve it  for several bipartite and multipartite scenarios by leveraging some structural properties. In doing so, we obtain thousands of new tight inequalities, many of which are also facets of the standard local polytope. We also define \emph{unanimous Bell inequalities}, a particular case of the previous class applied to the multipartite setting in which only full-equality events (all outcomes equal) are considered. We show that such inequalities can always be written as deterministic nonlocal games, and we give a simple multipartite unanimous family and prove its local bound. We show that most of these inequalities admit quantum violations, and we also display aspects of their importance for nonlocality. For instance, we identify examples where such inequalities can act as dimension witnesses, outcome witnesses, witnesses of genuine multipartite nonlocality, as well as being relevant to CHSH. These results show that these simple and elegant inequalities by themselves provide a powerful tool for discovering new Bell inequalities and device-independent witnesses.

\end{abstract}


\maketitle


{\em Introduction --- } In a Bell test, several distant parties perform measurements on subsystems of a composite physical system \cite{bell1964einstein,CHSH1969,Hensen2015}. Measurements have outcomes and in Bell scenarios we often use discrete variables to label them, independently of their nature. Here we define a particular type of Bell inequalities where arbitrary outcomes are directly compared. To illustrate our scenario, imagine the measurement outcomes are smells.\footnote{Note that measurement outcomes can be taken to be quite general. Here, we extend this idea to a regime in which outcomes need not correspond to physical systems, but may instead be perceptual stimuli, with smells serving as an illustrative example. We focus on smells because an individual smell may conceivably be difficult to identify, whereas determining whether two smells are identical may be comparatively easier.}\hspace{.05cm}Assume that the only process one can carry out is to compare if  such smells are identical or not. In such a case, the only probabilities that can enter the Bell inequality are of the form $p(=_{AB}|xy)\equiv p(a=b|x,y) = \sum_k p(k,k|x,y)$, where we restricted ourselves to two parties, Alice and Bob, with outcomes $a$ and $b$ and measurement settings $x$ and $y$, respectively. 
 
Most Bell inequalities cannot be expressed using only equalities between outcomes, for example, the \(I_{3322}\) and CGLMP \cite{I3322,CGLMP} inequalities. But some inequalities are suitable, like the famous CHSH-Bell inequality, which can be written as \footnote{Starting from the usual full-correlator version of CHSH and assuming normalization of the behavior.}:
\begin{equation}
     \label{CHSH}
     \begin{split}
         & p(=_{AB}|00)+p(=_{AB}|01) \\
 + & p(=_{AB}|10)-p(=_{AB}|11)\leq 2.     
     \end{split}
\end{equation}

Restricting the probabilities that can appear in a Bell inequality restricts the local polytope. Here, we analyze these restricted polytopes and classify all facet-defining \textit{Bell inequalities for smells} in some scenarios. Notably, we find that increasing the number of outcomes beyond certain \textit{saturated values} does not allow for new Bell inequalities. 

In the case of two parties and binary outcomes, the restricted polytope is identical to that of the XOR scenario (correlation polytope, no marginals) \cite{cleve2004consequences}. In fact, these inequalities can be seen as a natural generalization of XOR games beyond binary outcomes.
Interestingly, we also find a connection between our problem and numbers called Bell numbers,  due to another Bell, the mathematician Eric Temple Bell \cite{Bell1934,Bell1938}, and also known as Ramanujan numbers \cite{becker1948bell,Rankin_1961,Ramanujan1957}.


{\em Bell inequalities and polytope for smells --- } Consider $n$ parties, where at each round, each party receives one of  $m$ input settings, and outputs one of $k$ possible outcomes. Such a scenario will hereby be denoted by $(n,m,k)$. We define a Bell inequality for smells such that only the equality relations between the outcomes of the parties are taken into account. Specifically, we are only interested in whether their outputs are equal or not, for each configuration of inputs. 
  For a fixed input tuple $\mathbf{x} = (x_1,\dots,x_n)$, and any realized outcome $\mathbf{a}=(a_1,\dots,a_n)$, instead of the full set of probabilities $p(\mathbf{a}|\mathbf{x})$, we assume that only equality relations between outcomes can be observed in each run of the experiment, thus inducing an equivalence relation on the parties, \(i \sim j \iff a_i = a_j.
\) Equivalence classes of this relation correspond to groups of parties that obtained the same outcome. The pattern of equalities can thus be represented by a partition $\sigma$ of the set of parties $\{1,\dots,n\}$.\footnote{A partition of a set  $\{1,\dots,n\}$ is a collection of non-empty subsets, such that, subsets are pairwise disjoint and their union gives the entire set.}\hspace{.1cm}For a given $n$ the number of such partitions is the $n$-th Bell number $B_n$~\cite{Bell1934,becker1948bell},
\begin{equation}
  B_n = \sum_{k=0}^{n} \stirling{n}{k},
\end{equation}
where $\stirling{n}{k} = \sum_{l=0}^{k} \frac{(-1)^{k-l} \, l^n }{(k-l)! \, l!}$ the Stirling number of the second kind.

Let us denote by $\Pi_n$ the set of all partitions of $\{1,\dots,n\}$, and by $\sigma \in \Pi_n$ a specific partition. The probabilities accessible in this scenario are then
\begin{equation}
  p(=_{\sigma}|\mathbf{x}) := \Pr[\text{equality pattern } \sigma \text{ is observed} \mid \mathbf{x}],
\end{equation}
for each choice of inputs $\mathbf{x}$ and each partition $\sigma$. For fixed $\mathbf{x}$ the $B_n$ probabilities $p(=_{\sigma}|\mathbf{x})$ sum to one, so $(B_n - 1)$ of them are independent. By convention, we can thus disregard the partition where all outcomes are different. In scenario $(n,m,k)$ there are $m^n$ possible input tuples, and the dimension of the behavior space \(\mathcal{P}_\Sigma\) is therefore
\begin{equation}
  \dim \mathcal{P}_\Sigma = m^n (B_n - 1),
\end{equation}
which is typically much smaller than $m^n (k^n-1)$, that is, the dimension of the full behavior space in the corresponding $(n,m,k)$ scenario with fixed outcomes. Thus, a Bell inequality for smells is a linear functional on $\mathcal{P}_\Sigma$,
\begin{equation}
  I =  \sum_{\sigma \in \Pi_n} \sum_{\mathbf{x}} \alpha^{(\sigma)}_{\mathbf{x}}\, p(=_{\sigma}|\mathbf{x}),
  \label{eq:BI_qualities}
\end{equation}
together with a bound $L_\Sigma$ such that $I \le L_\Sigma$ for all classical behaviors.\footnote{Classical behaviors are the ones respecting Bell locality condition, such that they can be written as \(
    p(a,b|x,y)  = \int_{\Lambda} \Pi(\lambda) \, p_A(a|x,\lambda) \, p_B(b|y,\lambda) \, d\lambda\).}\hspace{.05cm}By restricting the full behavior $p(\mathbf{a}|\mathbf{x})$ to equality patterns one obtains a reduced behavior in $\mathcal{P}_\Sigma$. The convex hull of all classical behaviors in $\mathcal{P}_\Sigma$ defines the sub-polytope $\mathcal{L}_\Sigma$.

An important structural result is that for $(n,m,k)$ scenarios there exists a \emph{saturated} number of outcomes $k^*$ such that allowing more outcomes does not create new classical vertices in $\mathcal{L}_\Sigma$.

\begin{theorem}
\label{Th1}
Consider an $(n,m,k)$ Bell scenario. There exists a finite integer $k^*$ such that the local polytope $\mathcal{L}_\Sigma(n,m,k)$ coincides with $\mathcal{L}_\Sigma(n,m,k^*)$ for all $k \geq k^*$. The explicit expression is given by
\begin{equation}
  k^* =
  \begin{cases}
    \dfrac{n(m+1)}{2}, & n(m+1) \text{ even},\\[1ex]
    \dfrac{n(m+1)}{2} - \dfrac{1}{2}, & n(m+1) \text{ odd},
  \end{cases}
\end{equation}
that is,
\begin{equation}
  k^* = \left\lfloor \dfrac{n(m+1)}{2} \right\rfloor, 
\end{equation}
where   \(\lfloor . \rfloor\) is the floor function.
\end{theorem}

The proof is combinatorial, based on representing deterministic strategies as multipartite graphs whose connected components encode equality groups; we refer to the Appendix \ref{app:saturating_scenario} for  details. In the bipartite case the result simplifies:

\begin{corollary}
\label{cor:bipartite_saturated}
For any bipartite $(2,m,k)$ Bell scenario, the local polytope $\mathcal{L}_\Sigma(2,m,k)$ is saturated for $k=m+1$ outcomes. That is, $\mathcal{L}_\Sigma$ obtained from $(2,m,k)$ coincides with that from $(2,m,m+1)$ for all $k \ge m+1$.
\end{corollary}

Thus, in bipartite scenarios we can restrict attention to at most $m+1$ outcomes without loss of generality.

In Theorem \ref{Th2}, we provide an analytical expression for the total number of vertices in a bipartite Bell scenario for smells.

\begin{theorem}\label{Th2}
     $\mathcal{L}_\Sigma$ for a bipartite $(m_A,m_B,k,k)$ scenario has the following number of vertices
    \begin{equation}
    \sum_{\alpha=0}^{m_A}\sum_{\beta=0}^{m_B}\binom{m_A}{\alpha}\binom{m_B}{\beta} \sum_{j=0}^{k'} \stirling{m_A-\alpha}{j}\stirling{m_B-\beta}{j} j!
    \end{equation}
where \(k'= k-[\alpha>0]-[\beta>0]\). 
\end{theorem}
In the previous, $[\cdot]$ denotes Iverson brackets, which evaluate to 1 if the proposition inside the brackets is true and 0 otherwise, and $\stirling{k}{j}$ is the Stirling number of the second kind. The proof is given in Appendix \ref{app:unanimous_vertices}.

The local vertices of Bell scenarios for smells can be constructed from the vertices of the standard Bell scenario, which can be computationally expensive, as the number of vertices seems to grow even more quickly than in the Bell scenario for smells. We note that for bipartite scenarios, the local vertices of Bell scenarios for smells can be constructed directly. In Appendix \ref{app:unanimous_vertices}, we show how to construct vertices of \textit{unanimous inequalities} (defined below) in general multipartite scenarios, which in the bipartite case reduces to the bipartite smells expression of Theorem \ref{Th2}.

\emph{Bipartite inequalities---}
First, we focus on bipartite inequalities, that is, inequalities in scenarios of the form $(2,m,k)$. Therefore we may restrict ourselves to $k\leq m+1$, due to Theorem \ref{Th1}. Nevertheless, since the number of vertices grows exponentially as the inputs increase, we are not able to go beyond five inputs using facet enumeration.  In Appendix \ref{subsec:Bipartite} we present some of the features of these facets in Table \ref{tab:bipartite}, for the scenarios we studied. In general, the inequalities showed some interesting features, namely, a fair amount of facets of \(\mathcal{L}_{\Sigma}\) are also facets of the Bell local polytope \(\mathcal{L}\). In Appendix \ref{subsec:Bipartite} we show some representative examples of inequalities that are party-permutation-invariant (PPI) and can be leveraged for device-independent certification of various features, for example, dimensional witnesses. In the following sections we will highlight some other interesting bipartite inequalities and their uses.

An interesting general property proven for all bipartite  inequalities for smells, whether facet or lower-dimensional faces, is that their No-Signalling (NS) and  Signalling (S) bounds coincide. This follows from the proof in Appendix \ref{app:nosignalling}, which explicitly constructs the vertices of the NS polytope, for arbitrary bipartite Bell scenarios for smells. It also follows from that proof, that  bipartite NS polytopes for smells, akin to the case of the local polytope for smells, admit of a saturating scenario, albeit a much more severe one. In fact, $k^*_{NS} = 2$ for all bipartite scenarios for smells. 

Furthermore, another feature that we verify for all bipartite facet inequalities (except the trivial positivity facets) is that $L < Q < \text{NS} = \text{S}$, 
where $L, Q, \text{NS}, S$ denote the local, quantum, no-signalling and signalling bounds, respectively. We conjecture that this is true of all facet bipartite Bell inequalities for smells. This would imply that there are no trivial bipartite facet inequalities for smells \footnote{This does not hold in the multipartite case as shown explicitly by \(S_{222}\), a tripartite inequality  in \((3,2,2)\) for which \(L=S\).} and that they always have a strict local to quantum and quantum to no-signalling gap, while having no no-signalling to signalling gap.

\emph{Multipartite inequalities ---} Moving to multipartite inequalities for smells, we study various scenarios up to four-partite cases. 

 We first focus on the tripartite case with binary inputs, that is, on $(3,2,k)$, scenarios. From Theorem \ref{Th1} we can restrict the number of outcomes to $k\leq 4$ outcomes. In Appendix \ref{app:Multipartite} we present Table \ref{tab:scenario_3xy}, for the   scenarios we study.  Noticeably, a fair amount of the tripartite facets for smells are also facets of the local polytope, although none is PPI. 

 In scenario $(3,2,2)$, we find three classes of facets, none of which is a facet of the local polytope. One of these, \(S_{222}\), which we give in Eq.\ref{s222}, is a trivial inequality in the sense that \(L=Q=NS=S=1 \). 
 \begin{remark}
    Although \(S_{222}\) is not a facet of the local polytope,--- where being a facet would mean that its dimension would be $d_P-1$, with $d_P$ being the dimension of the polytope in question --- it is a ridge of the polytope, that is, it has dimension $d_P-2$. Indeed, although there are no known trivial facets \((L= Q = NS=S)\) of the local polytope, not counting positivity,  \(S_{222}\) shows that this is possible for a ridge of the local polytope to be trivial in this regard. 
 \end{remark}
 \(S_{222}\) will be further studied, since its output lifted version to \((3,2,3)\), will become non-trivial (that is, \(L<S\)) and exhibit quantum violation, as well as some other interesting properties.

 For the four-partite scenario we study  the simplest scenario, $(4,2,2)$. A summary of our results is shown in Table \ref{tab:scenario_4xy} in Appendix \ref{app:Multipartite}. We find 6.5 million classes of facets in this scenario. Once again, none are PPI but around 5 million of these classes are also facets of the standard local Bell polytope. This is noteworthy since not many facet Bell inequalities are known in \((4,2,2)\). Indeed, in \cite{bancal2010looking}, using a projection of the polytope to the PPI subspace, the authors managed to find 392 classes of such inequalities out of the total 627 facets of the symmetrized polytope.  
 
\emph{Unanimous Bell Inequalities --- }In the multipartite case, a particularly simple subclass of patterns are those in which \emph{all} parties obtain the same outcome; that is, the equality pattern for $n$ parties induces the  whole set $\{1,\dots,n\}$ as the partition for some inputs tuples \(\mathbf{x}\). We call this subclass \textit{Unanimous Bell Inequalities}, and they are written as
\begin{equation}
\label{eq:unanimous inequality}
  I_{\mathrm{un}} = \sum_{\mathbf{x}} \beta_{\mathbf{x}}\, p(=_{ABC...}|\mathbf{x}) \le L_{\mathrm{un}},
\end{equation}

where $L_{\mathrm{un}}$ is the local bound. The corresponding behavior space has dimension at most $m^n$, and the local polytope $\mathcal{L}_{\mathrm{un}}$ is the convex hull of all classical (unanimous) behaviors $ p(=_{ABC...}|\mathbf{x})$.

Unanimous inequalities can be regarded as a special case of inequalities for smells, in which only one partition (the fully-equal one) is retained for each input tuple. Indeed, in the bipartite  case both unanimous and smells characterizations coincide, since the fully-equal partition enforced by the unanimous definition is the only relevant partition in the smells bipartite  scenario. As such, unanimous inequalities inherit some of the structural properties of $\mathcal{L}_\Sigma$, but they exhibit an even stronger saturation phenomenon: for unanimous games the saturated number of outcomes depends only on the smallest number of inputs among the parties and is independent of $n$. 

In Appendix~\ref{app:unanimous_vertices} we give a construction for the local vertices of the \textit{unanimous inequalities} scenario, for any number of parties.  Using this explicit vertex construction we also show that $k^\ast = m_{\min} + 1$, where $m_{\min}$ is the minimum number of settings among the parties, giving for the bipartite case (where unanimous and inequalities for smells coincide) a direct proof of Corollary \ref{cor:bipartite_saturated}. We formalize this result in the following theorem: 

\begin{theorem} \label{Th3}
In scenario $(n,\mathbf{m},k)$, where $\mathbf{m} = (m_A,m_B,m_C,...)$ the local polytope $\mathcal{L}_{\mathrm{un}}(n,\mathbf{m},k)$ corresponding to unanimous games has $N_{\mathrm{un}}$ vertices, where
\begin{multline*}
 N_{\mathrm{un}} = \sum_{i_1=0}^{m_A} \sum_{i_2=0}^{m_B} \sum_{i_3=0}^{m_C}   \dots\binom{m_A}{i_1}\binom{m_B}{i_2} \binom{m_C}{i_3}\dots  \\ \sum_{j=0}^{k'}  \stirling{m_A-i_1}{j}\stirling{m_B-i_2}{j}\stirling{m_C-i_3}{j} \dots (j!)^{n-1}
\end{multline*}

where $k'= k- g(i_1,i_2,i_3,..)$  and

\begin{equation}
g(i_1,i_2,i_3,..) = 
     \begin{cases}
       0 &\quad\text{if} \;  i_l = 0,  \; \forall \, l \\
       2 &\quad\text{if} \; i_l > 0, \; \forall \, l \\
       1 &\quad \text{else}
     \end{cases}
\end{equation}

Noticeably, it follows that for a saturating value of $k^* = m+1$, one has $\mathcal{L}_{\mathrm{un}}(n,\mathbf{m},k) = \mathcal{L}_{\mathrm{un}}(n,\mathbf{m},k^*)$ for any number of outcomes $k \geq k^*$, where $m = \emph{min} \; \mathbf{m}$.

\end{theorem}

 Another interesting feature of unanimous inequalities is that  starting from Eq.\ref{eq:unanimous inequality}, one can always re-write the inequality as a deterministic nonlocal game \(G_{un} \),   that is, \( G_{un} = \sum_{\mathbf{a},\mathbf{x}} \mu(\mathbf{x})\, V(\mathbf{a|\mathbf{x}}) \, p(\mathbf{a} | \mathbf{x})\), with deterministic predicate \(V(\mathbf{a|\mathbf{x}})\in\{0,1\}\) and a valid probability distribution \(\mu(\mathbf{x})\). This result is formally stated in Theorem \ref{Theorem: Unanimous Games}, which along with its proof, is in Appendix \ref{section:All unanimous inequalities are nonlocal games}. Since for bipartite scenarios, unanimous inequalities and  inequalities for smells are equivalent, this also allows us to conclude that all bipartite inequalities for smells are deterministic nonlocal games.

As an example, we present a multipartite unanimous Bell family, in the scenario $(N,2,k)$, defined for any number of parties $N\geq 2$ and outcomes \(k\geq3\):
\begin{multline} \label{Unanimous Family}
\mathcal{F}_{N2}
= \\
\sum_{x_1,\dots,x_N}
(-1)^{x_1 \oplus \cdots \oplus x_N}\,
p\!\left(=_{1,\dots,N}\mid x_1,\dots,x_N\right) \\
\le 1 + (N+1)_{\mathrm{mod}\,2}.
\end{multline}

The local bound is $1$ if $N$ is odd, $2$ if $N$ is even. This inequality has a $+1$ coefficient on terms where the parity of the input bit string is $0$, and a $-1$ coefficient on terms where it is $1$. 
The proof for the local bound as a function of $N$ is given in Appendix \ref{app:local_bound_family}. We numerically find that for $N$ even, $\mathcal{F}_{N2}$ has its local bound equal to its no-signaling one, thus no quantum violation seems possible. However, for $N$ odd we do find examples of quantum violations. Examples of such quantum violations are shown at the end of Appendix \ref{app:local_bound_family}. The violations suggest that $\mathcal{F}_{N2}$ can act as a dimension witness. We conjecture that $\mathcal{F}_{N2}$ admits maximal violation using an $N-$qutrit state, for any $N$ odd. 

\emph{\(\mathcal{S}_{33}\): A noteworthy inequality in (2,3,3) ---}
\label{sec:s33}We now focus on an inequality (\(\mathcal{S}_{33}\)) which has some interesting features, namely, it serves as a dimension witness for states while simultaneously being \textit{relevant to CHSH} \cite{I3322,gisin2007bellinequalitiesquestionsanswers}, that is, it detects nonlocality in states that do not violate the CHSH local bound.   The inequality in question is, 
\begin{equation} \begin{split}
 \label{F3333}
  \mathcal{S}_{33} = p(=_{AB} | 00) + p(=_{AB}| 01) + p(=_{AB} | 02)\\ + p(=_{AB}| 10) - p(=_{AB}| 12) - p(=_{AB}| 20)\\ + p(=_{AB}| 21) - p(=_{AB}| 22) \le 3.
\end{split} \end{equation}

Given an inequality $I\le L$, we denote by $Q_d(I)$ the maximal quantum value when each party has local Hilbert space of dimension at most $d$.

A \emph{dimension witness} is a Bell inequality whose maximal quantum value depends on the dimension of the local Hilbert spaces.  If $Q_{d_0}(I) < Q_{d_1}(I)$ for some $d_1 > d_0$, then observing a violation larger than $Q_{d_0}(I)$ certifies that the local dimension must be at least $d_1$. 

Using the local polytopes we define above, one can compute the local bounds of this inequality $L_k$ for different values of the number of outcomes $k$, namely, 
\(L_2(\mathcal{S}_{33}) =  L_4(\mathcal{S}_{33}) = 3 =L_{\infty}(\mathcal{S}_{33})\). That is, the local bound using two outcomes already reaches its saturated value (i.e. for the scenario with $k=4$ outcomes, see Corollary \ref{cor:bipartite_saturated}).  On the other hand, by heuristically optimizing over quantum strategies we compute the following lower-bounds on the quantum bounds \(Q_d\), for local dimension $d$ of the bipartite shared state:
\begin{equation*}
\begin{split}
 &Q_2(\mathcal{S}_{33}) \geq 3.5;\;\\
 &Q_3(\mathcal{S}_{33}) \geq 3.6325;\;
Q_4(\mathcal{S}_{33}) \geq 3.6879.
\end{split}
\end{equation*}  Running the Navascués-Vértesi hierarchy over all possible configurations \cite{NV_Hierarchy_OG}, we find that $Q_2(\mathcal{S}_{33}) \leq 3.5$, that is, the lower-bound we find for this case is tight, and higher violations (using arbitrary projective measurements) ensure a larger Hilbert space dimension. Our lower-bounds for $Q_3(\mathcal{S}_{33})$ and $Q_4(\mathcal{S}_{33})$ also provide evidence of dimension witnessing using $\mathcal{S}_{33}$ for local dimension larger than $3$.

Regarding its relevance to CHSH, we find two-qubit entangled states which cannot violate CHSH (using arbitrary measurements) but violate \(\mathcal{S}_{33}\). For instance, \(\mathcal{S}_{33}\) detects the following rank-2 state 
\begin{equation}
    \rho_{p,\theta} = p \ketbra{\Psi_1}{\Psi_1} + (1-p) \ketbra{\Psi_2}{\Psi_2}
\end{equation}
where $\ket{\Psi_1} = \cos(\theta) \ket{00} + \sin(\theta) \ket{11} $, $\ket{\Psi_2} = \ket{01}$, and in the range $p = 0.955, \theta = \pi/14$ it turns out that $\rho_{p,\theta} $ is entangled (it has a concurrence $\approx 0.4144$) but cannot violate CHSH, as can be checked from the Horodecki criterion \cite{Horodecki_criterion}. However, we find that
\begin{align*}
    \mathcal{S}_{33}\left(\rho_{0.955,\pi/14}\right) \geq 3.0023 > 3,
\end{align*}

for suitable measurements (given as auxiliary files for reproducibility in the GitHub online repository \cite{git_unlabelled}).

Considering that these inequalities can be seen as a generalization of XOR inequalities beyond binary outcomes, it is interesting to note that already for trit outcomes, one manages to find an inequality that is \textit{relevant} to CHSH (which is also XOR), in such a low scenario. This seems noteworthy since, for bit outcomes, for which bipartite smells reduce to XOR inequalities,  finding examples of relevant inequalities to CHSH typically demands for significantly higher scenarios \cite{Vertesi08MoreEfficientBell,brierley2017convexseparationconvexoptimization,Designolle2023ImprovedLocalModels,Designolle2026BetterBounds}.

\emph{Other notable examples of inequalities and their properties ---} Below, we give a short explicit list of Bell inequalities for smells and unanimous inequalities which have remarkable and useful properties.

{\em Party permutation invariant facet of the local polytope and dimension witness}

In scenario $(2,4,5)$ we conjecture having solved the polytope $\mathcal{L}_{\Sigma}$ and find that $1253$ of its $3291$ facets were also facets of the local polytope (see Table \ref{tab:bipartite}). Moreover, we find $10$ of these facets of the local polytope to be party permutation  invariant (PPI), see Appendix \ref{subsec:Bipartite}. Here, we give one representative of this list, the penultimate inequality, \( \mathcal{S}_{4455}\), which displays dimension witness properties. Thus, 
\begin{equation}
\begin{split}
   & \mathcal{S}_{4455} = 3 p(=_{AB}|00) + 3 p(=_{AB}|01)  \\ & +2 p(=_{AB}|02) + p(=_{AB}|03) 
   + 3 p(=_{AB}|10)\\& -  p(=_{AB}|11)  - 2 p(=_{AB}|12) - 3 p(=_{AB}|13) \\&
  + 2 p(=_{AB}|20) - 2 p(=_{AB}|21)  - 3 p(=_{AB}|22) \\ &+ 2 p(=_{AB}|23) 
  + p(=_{AB}|30) - 3 p(=_{AB}|31)\\&  + 2 p(=_{AB}|32) + 2 p(=_{AB}|33) 
     \leq  12
\end{split}
\end{equation}
 is a PPI bipartite Bell inequality for smells, which is a facet of $\mathcal{L}_{\Sigma}$ as well as a facet of the local polytope. Defining $L_k$ as the local bound reached using $k$ outcomes we compute 
\(
L_2(\mathcal{S}_{4455})  = L_5(\mathcal{S}_{4455}) = 12 = L_{\infty}(\mathcal{S}_{4455}).
\)

That is, the local bound using two outcomes already reaches the value for the scenario in which the inequality was derived (i.e. for $5$ outcomes), which is the saturating scenario here (so the local bound cannot grow further). On the other hand, we compute the following quantum bounds \(Q_d\) (with local dimension $d$ of the bipartite shared state) :
\begin{equation}
\begin{split}
 & Q_2(\mathcal{S}_{4455}) \geq 13.08 , \\ 
& Q_3(\mathcal{S}_{4455}) \geq 13.84  , \; Q_4(\mathcal{S}_{4455}) \geq 14.17.
\end{split}
\end{equation}

Thus, $\mathcal{S}_{4455}$ seems to act as a dimension witness for dimensions between $2$ and $4$, supported by these explicit  lower bounds, coupled with the instances of tractable Navascués-Vértesi hierarchy that we run \cite{NV_Hierarchy_OG}, codes in \cite{git_unlabelled}. 

In Appendix \ref{subsec:Bipartite} we give some other examples of PPI inequalities for smells in different scenarios.

{\em Genuine multipartite nonlocality witness and outcome cardinality witness} 

In multipartite Bell scenarios one can distinguish between several inequivalent forms of multipartite nonlocality \cite{Svetlichny1987Multi}, the strongest notion of which is called genuine multipartite nonlocality (GMNL) \cite{Bancal2013GMNL}. 

One can detect GMNL by violating "sufficiently" a standard Bell inequality, i.e. violating it above the bilocal threshold \cite{Bancal2013GMNL}, which is the maximal value for which violations can be reproduced by a convex combination of bilocal local-hidden-variable models, see Appendix \ref{app:GMNL}. 

Since we consider multipartite scenarios, it is thus natural to ask whether some of our inequalities can demonstrate genuine multipartite nonlocality. Here, we give two explicit examples of such inequalities, a tripartite Bell inequality for smells and a four-partite unanimous inequality, which can detect genuine multipartite nonlocality, defined with no-signaling resources for the grouped parties. 

Consider the following Bell functional for smells, defined in the $(3,2,k)$ scenario:
\begin{equation}
\begin{split}
\label{s222}
& \mathcal{S}_{222} =   
  p(=_{AB|C}|110) + p(=_{AC|B}|001)  \\& - p(=_{AC|B}|010)  +  p(=_{AC|B}|100)   
 \\& - p(=_{AC|B}|101)  + p(=_{AC|B}|111)    \\
 &  + p(=_{BC|A}|011) .
\end{split}
\end{equation}

Again, we can compute the local bounds of this inequality $L_k$ for different values of the number of outcomes $k$. As we have already established the inequality is trivial in scenario \((3,2,2)\), so it is no surprise that \(L_2(\mathcal{S}_{222})= 1\), but now we have  
\begin{align}
L_3(\mathcal{S}_{222}) & =L_4(\mathcal{S}_{222}) = 2 = L_{\infty}(\mathcal{S}_{222}).
\end{align}

Here, the local bound using $2$ outcomes is smaller than the one using $3$ outcomes, which reaches its saturated value, i.e. the value for the scenario with $k=4$ outcomes, see Theorem \ref{Th1}.  The best quantum bounds $Q_d$ we find heuristically are
\begin{align}
  & Q_2(\mathcal{S}_{222}) = 2 \nonumber \\
  & Q_3(\mathcal{S}_{222}) =  Q_4(\mathcal{S}_{222}) \approx 2.000895  
\end{align}

Interestingly, all these strategies provably need at least $3$ outcomes (since for two outcomes the local bound equals the signaling one). These violations therefore act as ``outcome cardinality witnesses".

As stated, one can also compute the local bound for tripartite genuine nonlocality, using arbitrary no-signaling resources on grouped parties \cite{Bancal2013GMNL}:
\begin{equation}
L_{\text{bilocal-NS}}(\mathcal{S}_{222}) = 2
\end{equation}

That is, bilocal-NS strategies do not over-perform fully local ones for $\mathcal{S}_{222}$. The violation of this inequality thus implies genuine tripartite nonlocality.   

Finally, we move to the four-partite scenario and we present an example which is a facet of $\mathcal{L}_{\mathrm{un}}$, the unanimous polytope: 
\begin{equation}
\begin{split}
&\mathcal{U}_4 =  p(=_{ABCD}|0000) - p(=_{ABCD}|1100) \\&
 + p(=_{ABCD}|1101)+ p(=_{ABCD}|1110) \\&
 - p(=_{ABCD}|1111) \leq 1.
\end{split}
\end{equation}

This particularly simple inequality has also a local bound which coincides with the NS-bilocal bound, i.e. $L_{\text{bilocal-NS}}(\mathcal{U}_4) = 1$ (for any number of outcomes). Using a four-qubit state, we reach a quantum value of

\begin{equation}
    Q_2(\mathcal{U}_4) = 5/4 > 1
\end{equation}

That is, the inequality is violated, and also acts as a genuine multipartite nonlocality witness.

{\em Discussion ---}   We studied a class of Bell inequalities only admitting direct comparison of outcomes at each round. 
We defined the local polytope $\mathcal{L}_\Sigma$ corresponding to this class and  solved several bipartite and multipartite scenarios. This yielded large families of tight inequalities for smells, many of which are also facets of the standard local polytope $\mathcal{L}$. The reduced dimension of the behavior space and the smaller number of vertices of $\mathcal{L}_\Sigma$ make facet enumeration computationally tractable in scenarios where the full local polytope is out of reach. Facet enumeration of the local Bell polytope is limited to very simple scenarios and although a few methods have been proposed to go beyond direct facet enumeration \cite{cope2019bell,jesus2023tight}, Bell inequalities for smells seem to provide, in their own right, a useful way to produce large amounts of facet Bell inequalities of the local polytope. 

We further introduced unanimous Bell inequalities as an even more restrictive scenario, focusing only on full-equality events, and introduced a simple multipartite unanimous family.  We exhibited unanimous facets up to five parties, and identified examples that act as GMNL witnesses and dimension witnesses. Among the inequalities we find, several have particularly appealing properties: they are party-permutation invariant, have simple analytical forms, and serve as tight Bell inequalities, dimension witnesses, outcome witnesses, or GMNL witnesses.  We expect that these inequalities, and the remainder of facets we have obtained, will be useful in future studies of multipartite nonlocality and device-independent protocols.

\begin{acknowledgments}
The authors thank Sébastien Designolle for computational power, specifically helping us produce the list of Bell inequalities for smells $(2,5,2)$. The authors also thank Pavel Sekatski and Jef Pauwels for discussions and insightful comments. RF and E.Z.C. acknowledge funding by FCT/MCTES - Fundação para a Ciência e a Tecnologia (Portugal) - through national funds and when applicable co-funding by EU funds under the project UIDB/50008/2020. E.Z.C. also acknowledges funding by FCT through project 021.03707.CEECIND/CP1653/CT0002.
F.H. acknowledges funding by FCT/MCTES - Fundação para a Ciência e a Tecnologia (Portugal) through exploratory project 2023.12518.PEX. N.G. acknowledges support by the Swiss NCCR-SwissMap.
\end{acknowledgments}

\bibliography{common}

\appendix


\section{Local strategies in  Bell inequalities for smells and their saturating bound} \label{app:saturating_scenario}

In Bell scenarios for smells, given a particular choice of inputs, the only relevant information is knowing which parties produce equal outputs. As such, the exact values of the outputs $a_x$, $b_y$ etc.,  given for inputs $x$, $y$ etc. are irrelevant, and only the information established by equality relation \((=)\) matters to evaluate whether $a_x = b_y$ (similarly for the other parties). 

This equality relation \((=)\)  imposed over the output set, induces an equivalence relation (a binary relation that is reflexive, symmetric, and transitive) over the input set. Consider for simplicity the bipartite  case, where for inputs \(x,y\) Alice and Bob deterministically give outputs \(a_x,b_y\), then
 \(x \equiv y \iff a_x = b_y\), or alternatively
 \(x \not\equiv y \iff a_x \neq b_y.\) A natural way to represent these strategies is to use multipartite graphs, where each  partition of nodes represents the input set of the corresponding player, such that when two inputs $x$ and $y$ necessitate equal outputs, we connect such nodes with an edge. In this representation it is easy to see that  \(x \; R\; y \iff a_x = b_y\), or alternatively
 \(x \not R y \iff a_x \neq b_y,\) where \(R\) is the reachability relation between two vertices, which is also an equivalence relation. Thus sets of reachable nodes define a \textit{connected component}  of the multipartite graph (a subgraph of the multipartite graph where each  node of it is connected by a path),  demanding for those inputs equal outputs.  As an illustration in the bipartite  case, consider the following bipartite  graph

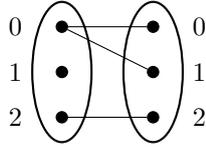
\begin{center}
\begin{tikzpicture}[baseline=(current bounding box.center), scale=0.6]

  \foreach \i in {0,1,2} {
    \node[draw, circle, fill=black, inner sep=1.5pt,
          label={[xshift=-9pt]left:${\i}$}] (a\i) at (0, -\i) {};
  }

  \foreach \j in {0,1,2} {
    \node[draw, circle, fill=black, inner sep=1.5pt,
          label={[xshift=9pt]right:${\j}$}] (b\j) at (2, -\j) {};
  }

  \draw[line width=0.8pt] (0, -1) ellipse [x radius=0.65, y radius=1.55];
  \draw[line width=0.8pt] (2, -1) ellipse [x radius=0.65, y radius=1.55];

  \draw (a0) -- (b0);
  \draw (a0) -- (b1);
  \draw (a2) -- (b2);

\end{tikzpicture}
\end{center}

where Alice's inputs are shown in the left partition, and Bob's in the right partition. This strategy, dictates that: 
\begin{itemize}
    \item Alice and Bob should output the same result when they receive either $0$  and $0$, or $0$ and $1$.
    \item Alice and Bob should output the same result when they  receive $2$ and  $2$;
    \item Alice and Bob should output different results in all other cases;
\end{itemize}

Note that by transitivity one could add an edge between $0$ and $1$  on Bob's side, since $a_0 = b_0$ and $a_0 = b_1$ implies \(b_0 = b_1\) but  edges within Alice's partition or Bob's partition are irrelevant, as what matters is the edges across partitions linking Alice's nodes and Bob's nodes, determining the grouping of the inputs into the connected components of the graph. The previous graph can be then partitioned into two distinct connected components and an isolated vertex on Alice's side. Thus, an implementation of this strategy would require that Alice has 3 possible symbols for outputs and Bob just 2.  In summary, this graph representation implies that only vertices belonging to the same bipartite connected component should have the same output.

\subsection{Proof of Theorem \ref{Th1}}

\subsubsection{Bipartite case}

Consider the bipartite graph \(G = (\{X, Y\}, E)\), where \(X,Y\) are the partitioned vertex sets, and \(E\) the set of \(X-Y\) edges. A bipartite connected component \(C \subseteq G\) is defined as a subset where any two nodes in the component are reachable from one another and not reachable from any node outside the component. By definition, all outputs within a single connected component must be the same. Additionally, outputs associated with different connected components must differ. Next, we introduce the sets:
\begin{align}
   \Tilde{V}_A = \{v \mid v \in X, \forall b \in Y, \{v, b\} \notin E(G)\} \\
\Tilde{V}_B = \{v \mid v \in Y, \forall a \in X, \{v, a\} \notin E(G)\}
\end{align}

These sets represent nodes that are not in relation to any node from the opposite partition. To minimize the total number of outputs, we assign a unique output to all nodes in $\tilde V_A$ and a different unique output to all nodes in $\tilde V_B$. Nodes in $\tilde V_A$ can share the same output value without creating additional equality relations, since none of them is connected to a node in the opposite partition. The same reasoning applies to $\tilde V_B$. \(\Tilde{V}_A\) and \(\Tilde{V}_B\) must have different outputs from each other and from any other bipartite connected component -otherwise, this would imply connectivity, leading to a contradiction. Consequently, any bipartite graph \(G = (\{X, Y\}, E)\) can be expressed as:
\[
G = C_1 \cup C_2 \cup \cdots \cup \Tilde{V}_A \cup \Tilde{V}_B
\]
Now, by requirement, we must assign a unique output label to each connected component \(C_i\) and  sets \(\Tilde{V}_A\) and \(\Tilde{V}_B\), provided they are non-empty. Therefore, the number of outputs required to cover all possible strategies is:
\[k^* = \max_{G}{\left\{|C| + [|\Tilde{V}_A| > 0] + [|\Tilde{V}_B| > 0]\right\}},
\] where \(C\) is the set of all  bipartite connected components of \(G\), and \([|\Tilde{V}_A| > 0]\) is the evaluation about whether the set is not empty in \(G\), in which case it will contribute with one distinct output, likewise for Bob.
Considering that the minimal connected component in a bipartite graph is a single \(X-Y\) edge, and there are at most as many connected components with single edges (call \(e\) the number of disjoint edges) as the number of nodes per party (or, in inhomogeneous scenarios, the number of nodes of the party with the smallest number of inputs),  we  find \(k^*\) by considering the following three possibilities:
\begin{enumerate}
    \item \(|\Tilde{V}_A| = |\Tilde{V}_B| = 0\), such that \(k_1^* = \max{\{e + 0\}} = m\), corresponding to no isolated node.
    \item \(|\Tilde{V}_A| = 0\), \(|\Tilde{V}_B| > 0\) (or vice versa), such that \(k_2^* = \max{\{e + 1\}} = m-1 +1 = m\), corresponding to one isolated node in one of the partitions.
    \item \(|\Tilde{V}_A| > 0\) and \(|\Tilde{V}_B| > 0\): \(k_3^* = \max{\{e + 2\}} = m-1 +2 = m + 1\), corresponding to one isolated node in both partitions.
\end{enumerate}

Thus, since the previous exhausts all possibilities, given a scenario for smells \((2, m, k)\), we have shown that no more than \(m + 1\) different outputs are needed to account for all possible local deterministic strategies.

\subsubsection{Multipartite case}

We can now extend a similar arguments for the generic multipartite case. Consider the (homogeneous) scenario $(n,m,k)$, where each of the $n$ parties have $m$ different inputs, each with $k$ outputs. We derive here the maximum value for $k$ such that all possible classical strategies are doable. 

We now have a multipartite graph $G$ consisting of $n$ partitions of $m$ nodes each, and a strategy consists again of a way to draw edges between nodes from different partitions, creating connected components corresponding to inputs having (deterministically) the same outputs. 

Like in the previous section, the maximum number of outputs one needs corresponds to the largest possible number of different groups over all graphs, which is given by the number of connected components plus some potential disconnected nodes: 

\begin{equation}
\begin{aligned}
k^* = \max_{G}\Bigl\{\, |C| + [|\Tilde{V}_A| > 0] + [|\Tilde{V}_B| > 0] \\
\qquad\qquad\quad + [|\Tilde{V}_C| > 0] \; + \, \ldots \Bigr\}
\end{aligned}
\end{equation}

where as in the previous section \(C\) is the set of all connected components of \(G\), and \([|\Tilde{V}_A| > 0]\) is the evaluation about whether the set for party $A$ is not empty in \(G\), in which case it will contribute with one distinct output, likewise for \([|\Tilde{V}_B| > 0]\), \([|\Tilde{V}_C| > 0]\), etc.

This number can be maximized in the following way: first note that, like in the bipartite case, the optimal number of disconnected nodes per party is $|\Tilde{V}| = 1$, for every party. That is, groups containing a single node are the most economical way to "use" an outcome. 
More formally, consider a graph with (say) party $A$ having no disconnected node. Taking one node and disconnecting it will either preserve the number of outcomes used (if the node was connected to a single different node and that party already had at least one isolated node) or add one outcome (in all other cases). On the other hand, having additional disconnected nodes does not add extra outcomes, as we saw that all disconnected nodes of one party can be allocated a single outcome. Hence, the graphs maximizing the number of groups can be taken to have

\begin{equation}
    |\Tilde{V}_A| = |\Tilde{V}_B| = |\Tilde{V}_C| = ... = 1
\end{equation}

Second, the next most economical groups are the connected components of exactly two nodes. We have $m-1$ nodes left per party, and we can directly upper-bound how many of these optimal groups one can make: 

\begin{align}
    |C| \leq \frac{n (m-1) }{2}
\end{align}

since this is the maximum number of pairs one can create with $n$ groups of $m-1$ nodes each, corresponding to perfect pair matching. 

The bound of perfect pair matching can in fact be reached whenever $n$ and $m - 1$ are not both odd, with the following recipe, using an arbitrary labeling for the $m - 1$ non-isolated nodes of each party:

\begin{equation}
\begin{minipage}{0.85\linewidth}
\begin{itemize}
    \item Take node $0$ from party $A$ and connect it to node $0$ of party $B$.
    \item Take node $1$ from party $B$ and connect it to node $0$ of party $C$.
    \item Take node $1$ from party $C$ and connect it to node $0$ of party $D$.
    \item \ldots
    \item Take node $1$ of the last party and connect it to node $1$ of party $A$.
\end{itemize}
\end{minipage}
\label{cyclic_procedure}
\end{equation}

This cyclic procedure creates a perfect pair matching for the first two nodes of each party. If $m-1$ is even, the procedure can be repeated until all nodes are part of an optimal $2-$node connected component, and one thus ends up with 

\begin{align}
     k^* &= |C| + [|\Tilde{V}_A| > 0] + [|\Tilde{V}_B| > 0] + [|\Tilde{V}_C| > 0] \; + \, ...   \nonumber \\
& = \frac{n (m-1) }{2} + n = \frac{n (m+1)}{2}
\end{align}

We illustrate the resulting graph for the case of $n=3$ parties with $m=3$ outputs on Figure \ref{fig:33}.

\begin{figure}[h]
  \centering
\input{Figures/tripartiteeven.tikz}
  \caption{A strategy maximizing the number of outputs for $3$ parties with $3$ inputs each, using cyclic procedure \ref{cyclic_procedure}. We end up with $3$ isolated nodes and a perfect matching of $3$ pairs for the remaining nodes, reaching $k^* = 6$. }
  \label{fig:33}
\end{figure}
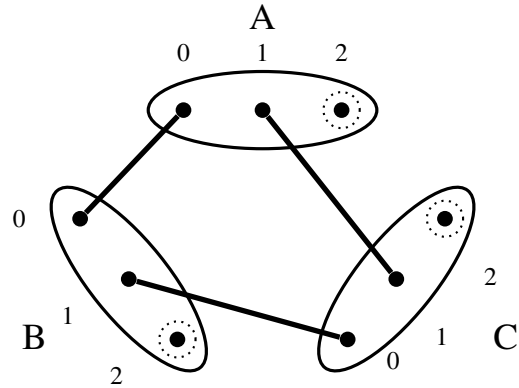

If $m-1$ is odd but $n$ is even one can also reach the bound: apply the cyclic procedure \ref{cyclic_procedure} for the first $m-2$ nodes of each party, achieving perfect pair matching for them. Next, there is exactly one node remaining per party, and since the number of parties is even one can match them perfectly by pairs of parties: $A - B$ , $C - D$, etc. One thus ends up again with the optimal case of $1$ disconnected node per party and all other nodes in optimal connected components of $2$ nodes. Again, 

\begin{align}
     k^* =  \frac{n (m-1) }{2} + n = \frac{n (m+1)}{2}
\end{align}

Finally, we are left with the case where both $n$ and $m-1$ are odd. In that case, the maximum number of pairs is given by $\frac{n (m-1) -1}{2} $ (since the number of nodes to connect, $n (m-1)$ is odd). This number can be achieved by the procedure described above: apply cyclic procedure \ref{cyclic_procedure} for the first $m-2$ nodes of each party, achieving perfect pair matching for them. Next, there is exactly one node remaining per party, and again we connect $A$ to $B$, $C$ to $D$, etc. Only one node of the last party will not be matched, thus achieving the optimal number of pairs for this scenario

\begin{align}
     k^* =  \frac{n (m-1) -1 }{2} + n = \frac{n (m+1)}{2} - \frac{1}{2}
\end{align}

We illustrate the resulting graph for the of $n=3$ parties and $m=4$ outputs.

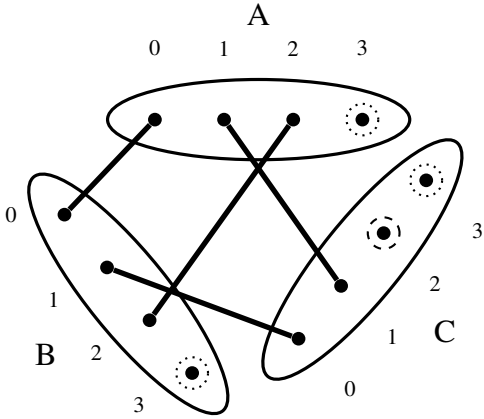
\begin{figure}[H]
  \centering
\input{Figures/tripartiteodd.tikz}
  \caption{A strategy maximizing the number of outputs for $3$ parties with $4$ inputs each, using cyclic procedure \ref{cyclic_procedure}. We end up with $4$ isolated nodes and a perfect matching of $4$ pairs for the remaining nodes, reaching $k^* = 7$. Note that the node with dashes around is ``unused'' due to the fact that both $n$ and $m-1$ are odd.}
  \label{fig:34}
\end{figure}

All in all, one thus obtains 

\[
k^* = 
\begin{cases} 
 \frac{n (m+1)}{2}, & \text{if } n (m+1) \text{ is even} \\
 \frac{n (m+1)}{2} - \frac{1}{2} , & \text{if } n (m+1) \text{ is odd}
\end{cases}
\]

\section{Construction of unanimous vertices} \label{app:unanimous_vertices}

For unanimous inequalities, which consider only cases where all parties produce the same output, we can explicitly construct the local deterministic strategies, for arbitrary number of parties and inputs. 

Let us assume $m_A$ measurements for Alice, $m_B$ for Bob, $m_C$ for Charlie, etc., and at most $k$ outcomes for each party. The general procedure to construct exhaustively the local deterministic strategies is the following: first pick the isolated nodes, then consider the number of unanimous components for the remaining connected nodes and construct the possible configurations for each number of unanimous component. We call \textit{unanimous component} a connected component with at least one node in each partition -as these are the only relevant connected components in the unanimous case. We detail the full procedure:

\begin{itemize}
    \item Determine which nodes will be the isolated nodes, which are here all nodes not belonging to a unanimous component

    \item Take the non-isolated nodes, and decide how many unanimous components there will be, that is, how many different outcomes will be used for these nodes

    \item Let us call the number of different outcomes $j$, what remains to do is to decide in which of the $j$ unanimous components each node goes. This determines the edges to draw (they link nodes which are in the same unanimous component) 
    
\end{itemize}

The first step will amount to

\begin{equation}
   \binom{m_A}{\alpha}\binom{m_B}{\beta} \binom{m_C}{\gamma} ...
\end{equation}

possibilities, where the multiplication runs over all parties, and $\alpha$ is the number of isolated nodes for Alice, $\beta$ for Bob, $\gamma$ for Charlie, etc.

Now, given $j$ outcomes (or equivalently $j$ unanimous components) in which one can sort the remaining  nodes we will have

\begin{equation}
    \stirling{m_A-\alpha}{j}\stirling{m_B-\beta}{j}\stirling{m_C-\gamma}{j} ...
\end{equation}

ways of partitioning the nodes of all parties (a partitioning of the nodes of a party into $j$ subgroups ensures that each of the $j$ outcomes is allocated at least one node of each party, while each local partition corresponds to a different strategy, given any fixed partitions for the other parties). However, this does not take into account the different "shuffling" of local subgroups (for example, the 'first' subgroup of Bob can still be allocated to any of the $j$ outcomes, the second subgroup to any of the remaining $j-1$, etc.). Thus, the overall number of strategies corresponding to having $j$ different unanimous components will be given by 

\begin{equation}
    \stirling{m_A-\alpha}{j}\stirling{m_B-\beta}{j}\stirling{m_C-\gamma}{j} ... (j!)^{n-1}
\end{equation}

where $n$ is the number of parties. The term $(j!)^{n-1}$ counts all the different ways of shuffling the $j$ subgroups for each party (except say Alice, whose partitioning defines the arbitrary allocation of subgroups to outcomes, only the other parties have to shuffle their subgroup-outcome matching). The case $j=0$ will account for exactly 1 configuration, based on the convention $\stirling{0}{0} = 1, \stirling{k}{0} = 0 , \forall k >0$  : the one where there is no connected component, i.e. the parties never give the same outputs.  

In overall, the number of assignments will therefore be given by

\begin{align}
\nonumber  \sum_{\alpha,\beta,\gamma,..=0}^{m_A,m_B,m_C,..}  \binom{m_A}{\alpha}\binom{m_B}{\beta} \binom{m_C}{\gamma} ... \\
\sum_{j=0}^{k - g(\alpha,\beta,\gamma,..)}  \stirling{m_A-\alpha}{j}\stirling{m_B-\beta}{j}\stirling{m_C-\gamma}{j} ... (j!)^{n-1}
\end{align}

where $k$ is the number of outcomes, and $g(\alpha,\beta,\gamma,..)$ is given by

\begin{equation}
g(\alpha,\beta,\gamma,..) = 
     \begin{cases}
       0 &\quad\text{if} \; \alpha = \beta = \gamma = ... = 0\\
       2 &\quad\text{if} \; \alpha > 0, \, \beta > 0, \, \gamma > 0 , ... \\ 
       1 &\quad  \text{else}  \\ 
     \end{cases}
\end{equation}

That is, the function $g(\alpha,\beta,\gamma,..)$ counts how many outcomes are already used for the isolated nodes (and thus have to be subtracted from the number of unanimous components available for the connected nodes). There are three cases: if there is no isolated node ($\alpha = \beta = \gamma = ... = 0$) all outcomes are available, $g =0$. If at least one node is isolated, but there is at least one party without an isolated node (e.g. $\alpha > 0 , \, \beta = 0$), one can attribute a single outcome to all the isolated nodes (since this will not result in a unanimous component), $g=1$. Finally, if all parties have at least one isolated node ($\alpha > 0, \, \beta > 0, \, \gamma > 0 , ...$), one has to use two different outcomes to associate to them, in a way that does not create a unanimous component (which can always be done by for instance associating outcome $0$ to all isolated nodes in the first $n-1$ parties, and outcome $1$ to the isolated node of party $n$), here $g=2$.

Finally, one can deduce from that formula that the number of vertices does not change from $k^* = m_A + 1$, where $m_A$ is the smallest number of inputs among the parties. Indeed, for $g=0$ (and thus $\alpha = 0$) the term

\begin{align}
 \stirling{m_A-\alpha}{j} = \stirling{m_A}{j}
\end{align}

will evaluate to $0$ for all $j > m_A$, and thus $k_{max} = m_A$. For $g=1$ and $\alpha=1$

\begin{align}
 \stirling{m_A-1}{j} 
\end{align}

will evaluate to $0$ for all $j > m_A - 1$ and thus $k_{max} - 1 = m_A - 1 \Rightarrow  k_{max} = m_A$. For $g=1$ and $\alpha=0$

\begin{align}
 \stirling{m_A}{j} 
\end{align}

will evaluate to $0$ for all $j > m_A $ and thus $k_{max} - 1 = m_A  \Rightarrow  k_{max} = m_A + 1$. Finally, for $g=2$ (and thus $\alpha \geq 1$)

\begin{align}
 \stirling{m_A-1}{j} 
\end{align}

will evaluate to $0$ for all $j > m_A - 1$ and thus $k_{max} - 2 = m_A - 1 \Rightarrow  k_{max} = m_A + 1$. In all cases, the sum evaluates to $0$ for $k > m_A + 1$.  

\section{Bell inequalities for smells: A partial repository}

Using Parallel AdjaceNcy Decomposition Algorithm (PANDA) \cite{lorwald2015panda}, a software for polyhedral transformation, we are able to enumerate the facets for various scenarios.  

{\em Symmetries--} Many of the produced facets will be equivalent under certain symmetries, so we must sort them into classes of symmetry. As in a standard Bell scenario, we may consider party, input, and output permutations. However, since we only score the players based on whether their outputs are equal, some output permutations will not have to be taken into account. Specifically, output permutations that preserve the probability distributions $p(k_{i_1}=\dots =k_{i_N})$ will result in exactly the same Bell inequality. Moreover, output permutations are not valid in generic scenarios, as they do not represent bijections anymore, i.e. they make vertices leave the probability subspace of the corresponding Bell scenario for smells. Output permutations will only be valid for scenarios with binary outcomes, for any number of parties. 

The party and the input permutations generate each a symmetry group which is isomorphic to the permutation group $\mathcal{S}_n$, which has two generators. Hence, to classify $(n,m,k)$ Bell inequalities for smells we only need to consider four generators for the whole symmetry group, except for $k=2$ where we add the output permutation.

\subsection{Bell facets for smells in $(2,m,k)$}
\label{subsec:Bipartite}

For each bipartite  scenario \((2,m,k)\) up to five inputs  (\(m=5\)) we provide in Table \ref{tab:bipartite}, the number of classes, the number of such classes which are also facets of the corresponding standard Bell scenario, and the number of party-permutation-invariant (PPI) inequalities. PPI Bell inequalities usually have simpler, elegant forms.

\begin{table}[t]
\caption{For each $(2,m,k)$ Bell scenario for smells, we provide the number of
classes, $\mathcal{L}$-facet classes, and PPI inequalities. We do not count
positivity in the classes.}
\begin{ruledtabular}
\begin{tabular}{lccc}
Scenario & Classes & $\mathcal{L}$ facets & PPI \\
\hline
(2,2,2) & 1 & 1 & 1 \\
(2,2,3) & 1 & 0 & 1 \\
(2,3,2) & 1 & 1 & 1 \\
(2,3,3) & 4 & 0 & 2 \\
(2,3,4) & 3 & 0 & 2 \\
(2,4,2) & 3 & 3 & 3 \\
(2,4,3) & 6947\footnote{\label{fn:conj}Scenarios for which we conjecture that the
list of classes of facets obtained is complete.} & 3277 & 163 \\
(2,4,4) & 3295\footnotemark[1] & 1255 & 86 \\
(2,4,5) & 3291\footnotemark[1] & 1253 & 84 \\
(2,5,2) & 1281\footnotemark[1] & 1281 & 154 \\
\end{tabular}
\end{ruledtabular}
\label{tab:bipartite}
\end{table}

In \((2,2,2)\), \((2,2,3)\), and \((2,3,2)\), we recover the CHSH inequality, as expected, and it turns out to be the only nontrivial class in these scenarios.  We conjecture that $(2,4,3)$, $(2,4,4)$ and $(2,4,5)$ are solved. Notice also that it naturally follows from Corollary \ref{cor:bipartite_saturated} that conjecturing that scenario $(2,4,5)$ is  solved lifts to a conjecture that $(2,4,k)$ for \(k\geq5\) is solved.

We focus on some PPI inequalities. For example, in \((2,3,3)\) we have inequality:
\begin{equation}
\begin{split}
    & -p(=_{AB} | 00) - p(=_{AB} | 01) - p(=_{AB} | 02)\\ & - p(=_{AB} | 10) 
     -p(=_{AB} | 11) + p(=_{AB} | 12)\\& - p(=_{AB} | 20) + p(=_{AB} | 21) 
     -p(=_{AB} | 22) \leq 1.
\end{split}
\end{equation}
Through heuristic methods (using seesaw of SDP's) we obtain the following  lower bounds for quantum violations \(Q_2\geq 1.24; Q_3 \geq 1.32\), furthermore, as with all bipartite  inequalities we studied we have \(NS=S\), which in this case equals 2.
Another PPI inequality with interesting properties, and possibly generalizable to larger scenarios, is the following, in the (2,4,3) scenario:
\begin{equation}
\begin{split}
    & p(=_{AB} | 00) + p(=_{AB} | 01) + p(=_{AB} | 02) - p(=_{AB} | 03) \\
    & + p(=_{AB} | 10) + p(=_{AB} | 11) - p(=_{AB} | 12)\\& + p(=_{AB} | 13) 
     + p(=_{AB} | 20) - p(=_{AB} | 21)\\& - p(=_{AB} | 22) - p(=_{AB} | 23)  - p(=_{AB} | 30) \\&+ p(=_{AB} | 31) - p(=_{AB} | 32) - p(=_{AB} | 33) \leq 4.
\end{split}
\end{equation}

For this one, we establish that \( 4.8284. \leq Q \leq 5.1\), where the lower-bound is obtained through an SDP seesaw analysis and the upper-bound  from NPA level ~2. We also have,  \(NS= S = 8\).

In $(2,4,5)$, from the \(3291\) classes of facets, we find that \(10\) of them are both facets of $\mathcal{L}$, and are PPI. The coefficients of these facets are listed in the  Table \ref{tab:matrix}. We focus on the penultimate one, \(\mathcal{S}_{4455}\) and studied it in more detail in the main text.

\begin{table*}
\caption{List of ten PPI Bell inequalities for the (2,4,5) scenario which are
facets of the smells or unanimous local polytope as well as facets of the (2,4,5)
standard Bell polytope. The first line gives the values of inputs $(x,y)$, with
corresponding coefficients $\beta_{xy}$ below.}
\begin{ruledtabular}
\begin{tabular}{rrrrrrrrrrrrrrrr}
00 & 01 & 02 & 03 & 10 & 11 & 12 & 13 & 20 & 21 & 22 & 23 & 30 & 31 & 32 & 33 \\
\hline
3 & 2 & $-1$ & $-1$ & 2 & $-2$ & 2 & $-2$ & $-1$ & 2 & 1 & 1 & $-1$ & $-2$ & 1 & 1 \\
4 & 2 & $-1$ & $-1$ & 2 & $-1$ & 1 & 1 & $-1$ & 1 & 1 & $-1$ & $-1$ & 1 & $-1$ & 0 \\
2 & 2 & 1 & 0 & 2 & $-3$ & $-2$ & 1 & 1 & $-2$ & 1 & $-1$ & 0 & 1 & $-1$ & 1 \\
6 & 2 & 1 & $-5$ & 2 & 1 & 2 & 3 & 1 & 2 & $-3$ & 2 & $-5$ & 3 & 2 & $-2$ \\
2 & 2 & 1 & 1 & 2 & $-1$ & $-1$ & $-2$ & 1 & $-1$ & $-2$ & 1 & 1 & $-2$ & 1 & 2 \\
6 & 2 & 1 & $-3$ & 2 & 4 & $-2$ & 4 & 1 & $-2$ & $-1$ & 1 & $-3$ & 4 & 1 & $-3$ \\
7 & 4 & 2 & 1 & 4 & $-3$ & $-4$ & $-2$ & 2 & $-4$ & 1 & 1 & 1 & $-2$ & 1 & $-2$ \\
8 & 5 & 2 & 1 & 5 & $-4$ & $-4$ & $-2$ & 2 & $-4$ & 1 & 1 & 1 & $-2$ & 1 & $-1$ \\
3 & 3 & 2 & 1 & 3 & $-1$ & $-2$ & $-3$ & 2 & $-2$ & $-3$ & 2 & 1 & $-3$ & 2 & 2 \\
5 & 4 & 3 & 1 & 4 & $-6$ & $-4$ & 2 & 3 & $-4$ & 1 & $-3$ & 1 & 2 & $-3$ & 2 \\
\end{tabular}
\end{ruledtabular}
\label{tab:matrix}
\end{table*}

\subsection{Multipartite Bell facets for smells}

\label{app:Multipartite}

{\em Bell tripartite facets for smells ---} We now turn to the tripartite case, specifically for binary inputs, that is, the scenario  $(3,2,k)$, where we can restrict the number of outcomes to $k\leq 4$ outcomes. Our findings are shown in Table \ref{tab:scenario_3xy}

\begin{table}[H]
\caption{For each $(3,m,k)$ Bell scenario for smells, we provide the number of
classes, $\mathcal{L}$-facet classes, and PPI inequalities.}
\begin{ruledtabular}
\begin{tabular}{lccc}
Scenario & Classes & $\mathcal{L}$ facets & PPI \\
\hline
(3,2,2) & 3 & 0 & 0 \\
(3,2,3) & 5731 & 729 & 0 \\
(3,2,4) & 555 & 349 & 0 \\
\end{tabular}
\end{ruledtabular}
    \label{tab:scenario_3xy}
\end{table}

In \((3,2,3)\), we also have the following facet,
\begin{equation}
\begin{split}
& \mathcal{S}_{222,333} =   \\&
  - p(=_{AB|C}|111) - p(=_{AC|B}|000) + p(=_{AC|B}|001) \\& +  p(=_{AC|B}|010)   
  - p(=_{AC|B}|011)  + p(=_{AC|B}|100) \\& - p(=_{AC|B}|101)  - p(=_{BC|A}|111) \leq  0,
\end{split}
\end{equation}
This inequality is also a facet of $\mathcal{L}$ and is not PPI. $Q_3 \geq 0.4574$ (for qutrits),  $NS= S = 1$. 


{\em Bell four-partite facets for smells.--} We now turn to the fourpartite Bell scenario for smells, where the saturating number of outcomes, as per Theorem \ref{Th1}, is $2(m+1)$. We focus exclusively on the simplest four-partite scenario, $(4,2,2)$, with  binary outputs, which is significantly  lower than the saturating number of outcomes for this case, i.e. \(6\). We find 6.5 million classes of facets in this scenario. Interestingly, around 5 million of these classes are also facets of the standard local Bell polytope.

\begin{table}
\caption{For the (4,2,2) Bell scenario for smells, we provide the number of
classes, $\mathcal{L}$-facet classes, and PPI inequalities.}
\begin{ruledtabular}
\begin{tabular}{lccc}
Scenario & Classes & $\mathcal{L}$ facets & PPI \\
\hline
(4,2,2) & $\sim\!6.5\times10^{6}$ & $\sim\!5\times10^{6}$ & 0 \\
\end{tabular}
\end{ruledtabular}
\label{tab:scenario_4xy}
\end{table}

\subsection{Multipartite Unanimous Bell Inequalities} \label{app:unanimous_facets}
 Since the saturation of the number of vertices is an interesting and useful structural property, and unanimous inequalities are always saturated in \((n,m,m+1)\), one can leverage this property to try to find Bell inequalities in scenarios with large numbers of settings and outcomes. Since in the bipartite case the scenario is equivalent to the Bell inequalities for smells, we draw our focus to the multipartite case. We find unanimous Bell inequalities up to scenarios with five parties. The drawback is that finding unanimous facets seems to be a less promising strategy for finding facets of the standard local polytope  \(\mathcal{L}\) in the standard scenario (compared to taking all partitions as in multipartite scenarios for smells).

Since, in general a facet of the Bell polytope for smells is not a facet of the corresponding local polytope, and this happens to be even truer for unanimous inequalities, we will want to quantify ''how close'' an inequality is to being a facet. We do so by employing the notion of \textit{facetness}, which for a given inequality is the dimension of the space engendered by its saturating vertices divided by the dimension of that space for a facet (that is, $d_P-1$, where $d_P$ is the dimension of the polytope in question). For example, an inequality saturated by a single vertex has facetness $0$ while for a facet-defining inequality the facetness is $1$.

The results are presented in Table \ref{tab:scenario_3xy_extended}.

\begin{table}
\caption{For each $(n,m,k)$ in the unanimous scenario, we provide the number of
classes, $\mathcal{L}$-facet classes, PPI inequalities, and maximum facetness we
find in the scenario. We do not count positivity in the classes.}
\begin{ruledtabular}
\begin{tabular}{lcccc}
Scenario & Classes & $\mathcal{L}$ facets & Maximum facetness & PPI \\
\hline
(3,2,2) & 3 & 0 & 19/26 & 1 \\
(3,2,3) & 7 & 0 & 115/123 & 2 \\
(3,3,2) & 27190\footnote{Scenarios for which we conjecture that the list of
classes of facets obtained is complete.} & 0 & 57/62 & 1 \\
(4,2,2) & 371 & 3 & 79/79 & 12 \\
(4,2,3) & 583 & 0 & 567/623 & 9 \\
(5,2,2) & $\geqslant 3\times10^{6}$ & 0 & 174/241 & 5 \\
\end{tabular}
\end{ruledtabular}
\label{tab:scenario_3xy_extended}
\end{table}

Out of the previous, we pick some interesting examples of unanimous facets. Namely, in the \((4,2,2)\) scenario, we find the only three unanimous facets which also happen to be facets of the local polytope. They are represented in Table \ref{tab:matrix2}.

\begin{table*}
\caption{Four-partite unanimous facets in scenario $(4,2,2)$. The first line is
the values of the four inputs $x,y,z,w$, with corresponding coefficients
$\beta_{xyzw}$ below (each line thus defining one inequality). These are the only
ones out of the 371 facets we find in this scenario which also are facets of the
local polytope $\mathcal{L}_{\Sigma}$.}
\begin{ruledtabular}
\begin{tabular}{rrrrrrrrrrrrrrrr}
0000 & 0001 & 0010 & 0011 & 0100 & 0101 & 0110 & 0111 & 1000 & 1001 & 1010 & 1011 & 1100 & 1101 & 1110 & 1111 \\
\hline
1 & $-1$ & $-1$ & 1 & $-1$ & 1 & 1 & $-1$ & $-1$ & $-1$ & $-1$ & $-1$ & $-1$ & $-1$ & $-1$ & $-1$ \\
1 & $-1$ & $-1$ & 1 & $-1$ & $-1$ & $-1$ & $-1$ & $-1$ & $-1$ & $-1$ & $-1$ & $-1$ & 1 & 1 & $-1$ \\
1 & $-1$ & $-1$ & 1 & $-1$ & 1 & $-1$ & $-1$ & $-1$ & 1 & $-1$ & $-1$ & $-1$ & $-1$ & 1 & $-1$ \\
\end{tabular}
\end{ruledtabular}
\label{tab:matrix2}
\end{table*}

\section{Genuine multipartite nonlocality} \label{app:GMNL}

Here we focus on the definition of genuine multipartite nonlocality (GMNL). Let us take the example of $n=3$ parties, the straightforward extension of the local polytope is as follows: 

\begin{equation}
    \begin{aligned}
            & p_L(a,b,c|x,y,z) = \\
    & \int_{\Lambda} \Pi(\lambda) p_A(a|x,\lambda) p_B(b|y,\lambda) p_C(c|z,\lambda) d\lambda
    \end{aligned}
\end{equation}

where $\lambda$ is the local variable (which follows distribution $\Pi(\lambda)$), $p_A(a|x,\lambda), p_B(b|y,\lambda), p_C(c|z,\lambda)$ the local response functions. (the extension to more than 3 parties is straightforward). In that case, behavior $p_L(a,b,c|x,y,z)$ is called fully tripartite local.

But different structures can be considered, for instance 

\begin{equation}
    p_{AB|C} = \int_{\Lambda} \Pi(\lambda)  p_{AB}(a,b|x,y,\lambda) p_C(c|z,\lambda) d\lambda
\end{equation}

where Alice and Bob can act together, so it corresponds to the bipartition $A B$ vs $C$. More generally, we can consider partitions where the grouped parties have access to different resources \cite{Svetlichny1987Multi,Bancal2013GMNL}.

If we consider convex combinations of all different bipartitions, we get the definition of genuine multipartite nonlocality, which in the example of 3 parties is a behavior $p(a,b,c|x,y,z)$ which cannot be written as

\begin{align*}
    p_{\text{bilocal}} = q_1 \int_{\Lambda} \Pi_1(\lambda)  p_{AB}(a,b|x,y,\lambda) p_C(c|z,\lambda) d\lambda \\
    + q_2 \int_{\Lambda} \Pi_2(\lambda)  p_{AC}(a,c|x,z,\lambda) p_B(b|y,\lambda) d\lambda  \\
    + q_3 \int_{\Lambda} \Pi_3(\lambda)  p_{BC}(b,c|y,z,\lambda) p_A(a|x,\lambda) d\lambda  
\end{align*}

 and one can consider different resources for the grouped parties, that is, $p_{AB}(a,b|x,y,\lambda)$ can be imposed to be quantum behaviors (for all $\lambda$), or no-signaling ones, or a one-way no-signaling, etc. In this work, we considered all grouped resources to be arbitrary no-signaling behaviors. That is, we compute the bound for GMNL $L_{\text{bilocal-NS}}$ by optimizing over all behaviors of the form 

\begin{align*}
    p_{\text{bilocal}} = q_1 \int_{\Lambda} \Pi_1(\lambda)  p_{AB}^{(NS)}(a,b|x,y,\lambda) p_C(c|z,\lambda) d\lambda \\
    + q_2 \int_{\Lambda} \Pi_2(\lambda)  p_{AC}^{(NS)}(a,c|x,z,\lambda) p_B(b|y,\lambda) d\lambda  \\
    + q_3 \int_{\Lambda} \Pi_3(\lambda)  p_{BC}^{(NS)}(b,c|y,z,\lambda) p_A(a|x,\lambda) d\lambda  
\end{align*}

\section{All unanimous inequalities are nonlocal games}
\label{section:All unanimous inequalities are nonlocal games}

Here we show that unanimous Bell inequalities can always be seen as deterministic nonlocal games. We formalize this in the following theorem:

\begin{theorem}
    
\label{Theorem: Unanimous Games}

Let $ I_{\mathrm{un}}$ be an arbitrary unanimous Bell inequality:

\begin{equation}
  I_{\mathrm{un}} = \sum_{\mathbf{x}} \beta_{\mathbf{x}}\, p(=_{ABC...} | \mathbf{x}) \le L_{\mathrm{un}},
\end{equation}

The following (standard) Bell inequality 

\begin{equation}
\begin{split}
      G_{\mathrm{un}} = \sum_{\mathbf{a},\mathbf{x}} \mu(\mathbf{x})\, V(\mathbf{a|\mathbf{x}}) \, p(\mathbf{a} | \mathbf{\mathbf{x}}) \\ \le \frac{1}{\sum_{\mathbf{x}} \abs{\beta_{\mathbf{x}}}} ( L_{\mathrm{un}} - \sum_{\mathbf{x}|\beta_\mathbf{x}<0} \beta_x )
    \end{split}
\end{equation}

defines a deterministic nonlocal game, equivalent to unanimous Bell inequality $I_{\mathrm{un}}$, and with prior $\mu(\mathbf{x}) = \frac{\abs{\beta_{\mathbf{x}}}}{\sum_{\mathbf{x}} \abs{\beta_{\mathbf{x}}}}$, and predicate 

\[
V(\mathbf{a},\mathbf{x}) = 
    \begin{cases}
       \delta_{a_1 , a_2 , a_3, ... }   &  \mathbf{x} |\beta_{\mathbf{x}} \geq 0 \\
       1 -   \delta_{a_1 , a_2 , a_3, ... }  & \mathbf{x} | \beta_{\mathbf{x}} < 0
     \end{cases}
\]

\end{theorem}

\begin{proof}

A (multipartite) deterministic nonlocal game is a functional $G(\mathbf{a},\mathbf{x)}$ that can be decomposed as 

\[
G(\mathbf{a},\mathbf{x)} = \mu(\mathbf{x}) \, V(\mathbf{a},\mathbf{x})
\]

with $\mu(\mathbf{x}) \geq 0$ a distribution over the set of inputs, thus satisfying $\sum_{\mathbf{x}} \mu(\mathbf{x}) = 1$, and $V(\mathbf{a},\mathbf{x}) \in \{0,1\}$. 

The corresponding Bell functional

\begin{equation} \label{Game_inequality}
  I_{G}(p)= \sum_{\mathbf{x}} G(\mathbf{a},\mathbf{x)} \, p(\mathbf{a} | \mathbf{x} )
\end{equation}

can be seen as a game where players receive input configuration $\mathbf{x}$ with probability $\mu(\mathbf{x)}$, and win (score $1$ point) if they answer some $\mathbf{a}$ such that $V(\mathbf{a},\mathbf{x}) = 1$, lose otherwise (score $0$ point). $\mu(\mathbf{x)}$ is called the prior distribution of the inputs, and $V(\mathbf{a},\mathbf{x})$ is called the predicate of the game. Equation \eqref{Game_inequality} is thus the average score at that game for players following a strategy corresponding to behavior $p(\mathbf{a} | \mathbf{x} )$. 

Deterministic nonlocal games are an interesting and simple subclass of generic nonlocal games, which constitute another (less geometrical) approach for Bell nonlocality not based on characterizing separating hyperplanes \cite{araujo2020bell}. 

In order to transform a generic unanimous inequality into a deterministic game, we write all terms $p(=_{ABC...} | \mathbf{x})$ which have negative coefficients in the inequality as $1 - p( \, \neq_{ABC...}| \mathbf{x})$, leading to an inequality with only positive coefficients, and a global shift of $ \sum_{\mathbf{x}|\beta_\mathbf{x}<0} \beta_x$:

\begin{equation}
\begin{split}
 & I_{\mathrm{un}} = 
 \\ & \sum_{\mathbf{x}| \beta_{\mathbf{x}} > 0}  \beta_{\mathbf{x}} \, p(=_{ABC...} | \mathbf{x}) +  \sum_{\mathbf{x}| \beta_{\mathbf{x}} < 0} \beta_{\mathbf{x}} \, (1 - p( \,\neq_{ABC...} | \mathbf{x}))   
 \\ & = \sum_{\mathbf{x}| \beta_{\mathbf{x}} > 0}  \beta_{\mathbf{x}} \, p(=_{ABC...} | \mathbf{x}) +  \sum_{\mathbf{x}| \beta_{\mathbf{x}} < 0} \abs{\beta_{\mathbf{x}}} \,  p( \,\neq_{ABC...} | \mathbf{x})   
 \\ & + \sum_{\mathbf{x}|\beta_\mathbf{x}<0} \beta_x
  \end{split}
\end{equation}

and thus we can define the "shifted" inequality $I'_{\mathrm{un}}$ as

\begin{equation}
\begin{split}
 &I'_{\mathrm{un}} = I_{\mathrm{un}} -  \sum_{\mathbf{x}|\beta_\mathbf{x}<0} \beta_x= \\
 & \sum_{\mathbf{x}| \beta_{\mathbf{x}} > 0}  \beta_{\mathbf{x}} \, p(=_{ABC...} | \mathbf{x}) +  \sum_{\mathbf{x}| \beta_{\mathbf{x}} < 0} \abs{\beta_{\mathbf{x}}} \,  p( \,\neq_{ABC...} | \mathbf{x})   
  \end{split}
\end{equation}

Now we can see this inequality $I'_{\mathrm{un}}$ as a game where for each input configuration the players can score (proportionally to) $\beta_{\mathbf{x}}$ if either (1) $\beta_{\mathbf{x}} > 0$ and they all output the same outcome (i.e. $a_1 = a_2 = a_3...$), or (2) $\beta_{\mathbf{x}} < 0$ and at least two players output different outcomes (i.e. $ \neg (a_1 = a_2 = a_3 =....)$); they score $0$ in all other cases. That is, the winning condition on inputs corresponding to positive coefficients is that all parties output the same outcomes, and its negation for inputs corresponding to negative coefficients:

\[
V(\mathbf{a},\mathbf{x}) = 
    \begin{cases}
       \delta_{a_1 , a_2 , a_3, ... }   &  \mathbf{x} |\beta_{\mathbf{x}} \geq 0 \\
       1 -   \delta_{a_1 , a_2 , a_3, ... }  & \mathbf{x} | \beta_{\mathbf{x}} < 0
     \end{cases}
\]

We can now divide $I'_{\mathrm{un}}$ by a well-chosen constant to make the resulting functional a valid deterministic game : 

\begin{equation}
 \begin{split}     
\frac{I'_{\mathrm{un}} }{\sum_{\mathbf{x}} \abs{\beta_{\mathbf{x}}}}   &=
 \frac{1}{\sum_{\mathbf{x}} \abs{\beta_{\mathbf{x}}}} \Big(\sum_{\mathbf{x}| \beta_{\mathbf{x}} > 0}  \beta_{\mathbf{x}} \, p(=_{ABC...} | \mathbf{x})   \\ &+\sum_{\mathbf{x}| \beta_{\mathbf{x}} < 0} \abs{\beta_{\mathbf{x}}} \, p( \, \neq_{ABC...} | \mathbf{x}) \,\Big) \\
&= \sum_{\mathbf{a},\mathbf{x}}  \frac{\abs{\beta_{\mathbf{x}}}}{\sum_{\mathbf{x}} \abs{\beta_{\mathbf{x}}}} \, V(\mathbf{a}|\mathbf{x}) \, p(\mathbf{a} | \mathbf{x} ) \\
&= \sum_{\mathbf{a},\mathbf{x}}  \mu(\mathbf{x}) \, V(\mathbf{a}|\mathbf{x}) \, p(\mathbf{a} | \mathbf{x} )
\end{split}
\end{equation}

 where by construction $\mu(\mathbf{x}) = \frac{\abs{\beta_{\mathbf{x}}}}{\sum_{\mathbf{x}} \abs{\beta_{\mathbf{x}}}}$ is positive and normalized, defining a valid deterministic nonlocal game. 

 Thus, an arbitrary unanimous inequality is transformed into a deterministic nonlocal game by adding a constant and then renormalizing it: 

\[
I_{\mathrm{un}} \rightarrow   \frac{1}{\sum_{\mathbf{x}} \abs{\beta_{\mathbf{x}}}}  (I_{\mathrm{un}} - \sum_{\mathbf{x}|\beta_\mathbf{x}<0} \beta_x )
\]

since adding arbitrary constants and multiplying by a non-zero scalar does not change the property of an inequality (it only changes the local bound), we see that all unanimous inequalities correspond to deterministic nonlocal games. 

\end{proof}

\section{Unanimous multipartite family: \\ Proof of local bound and some quantum violations}
\label{app:local_bound_family}

\begin{theorem}
\label{Family Game Form}
   Call \(G_{N2}\) an \(N-\)partite family of nonlocal games in the scenario \((N,2,3)\), that is, having input \(\textbf{s}= (s_1, ..., s_N)\in \{0,1\}^N\), output \(\textbf{k}=( k_1, ..., k_N)\in \{0,1,2\}^N\) such that, 
    
    \[\mu(\textbf{s})= \frac{1}{2^N};\]
    \[ V(\textbf{k}\;|\textbf{s}) = \left[\bigoplus_{i=1}^n s_i = 1- \delta(k_1, ... , k_N)\right].\]

    The classical value for this family as a function of \(N\) is  for \(E(N) := [N \equiv 0 \pmod{2}]\)\[\omega_c(N) = \frac{1}{2}+\frac{2^{E(N)}}{2^N}.\]
\end{theorem}
\begin{proof}First, let us state that from Fine's theorem deterministic assignments from inputs to outputs are sufficient to achieve the local bound. 
 For any integer $N\ge 2$  define the \(N\)-bit string $s=(s_1,\dots,s_N)\in\{0,1\}^N$. For each index $i\in[N]$ specifying the bit \(s_i\) in the string, let $f_i:\{0,1\}\to\{0,1,2\}$, be a map that assigns to each bit-value a trit-value, such that, the output \(\textbf{k}= (k_1, ..., k_N) = (f_1(s_1), ..., f_N(s_N))\). Let us define the following string subsets:

\begin{equation*}
    \begin{split}
        &E:=\{ s \;| \ \bigoplus_i s_i = 0\} \subset \{0,1\}^N. \\
        & O:=\{s \;|\ \bigoplus_i s_i = 1\} \subset \{0,1\}^N.
    \end{split}
\end{equation*}
i.e. $E$ (resp.\ $O$) is the set of even (resp.\ odd) parity bitstrings. Define also the following sets for \({t}\in \{0,1,2\}\)
\[
A_t:=\{s\;|\;f_1(s_1)=f_2(s_2)=\cdots=f_N(s_N) = t\}.
\] such that  \( A = A_0 \cup A_1 \cup A_2\). 

In set notation, the predicate demands that \( s\in E \Leftrightarrow s\in A \Leftrightarrow s\in E \cap A\)  or alternatively (by its contrapositive), \( s\in O \Leftrightarrow s\notin A \Leftrightarrow s \in O\setminus A \). Then define \(S\) as the number of strings satisfying  either, that is, 
\begin{equation*}
\begin{split}
   S:=&|E\cap A|\;+\;|O\setminus A| \\ =& 2^{N-1}+\bigl(|E\cap A|-|O\cap A|\bigr)
=2^{N-1}+D.
\end{split}
\end{equation*}

Since the game is defined for uniform inputs, we have
\[\omega_c(N) = \frac{S}{2^N}= \frac{2^{N-1}+D}{2^N} =  \frac{1}{2}+\frac{D}{2^N}.\]

We will bound \(D\) as shown, 

\[
 D \le \left\{
\begin{array}{l}
 1,\; \text{for}\; N\text{-odd} \\[2pt]
2,\; \text{for}\; N\text{-even}
\end{array}
\right.\]
and give a construction to show that the bounds are sharp.

 A trit $t\in\{0,1,2\}$ is constant at coordinate $i$ if $f_i(0)=f_i(1)=t$, and if a trit is constant in at least one coordinate $i$, we say that it is constant in the string. Given a string, there will necessarily exist   \(3,2,1\) or \(0\) constant trits in the string.  We will bound \(D\) for all of the cases which exhaust all possibilities.

\textit{Case 1: Two or three trits constant in the string}\\
Suppose there exist at least two distinct constant trits in the string,  $t\neq t'$ and coordinates $i\neq j$ with $f_i(0)=f_i(1)=t$ and $f_j(0)=f_j(1)=t'$. Then for every $s$ the outputs at $i$ and $j$ are
$t$ and $t'$ respectively, so no $s$ can have all outputs equal. Hence $A=\varnothing$, so
$D=0$ and $S=2^{N-1}$.

\textit{Case 2: One trit constant in the string}\\
Assume there is a constant trit $t^*$ for coordinates $i\in I$ with $f_i(0)=f_i(1)=t^\ast$, and no other trit
is constant anywhere. Then any $s\in A$ must have common value $t^\ast$, and $s\in A\Leftrightarrow s\oplus 1_i\in A$ (flipping the $i$-th bit preserves all outputs). Thus $A$ is a disjoint union of
pairs $\{s,s\oplus 1_i\}$ for all \(i\in I\) corresponding to one even and one odd string, so $|E\cap A|=|O\cap A|$ and
$D=0$. Therefore $S=2^{N-1}$.

\textit{Case 3: No trit constant in the string}\\
In this case for every coordinate $i$ we have $f_i(0)\neq f_i(1)$, so that if $s$ belongs to $A_t$,  necessarily $t\in\{f_i(0),f_i(1)\}$ for all $i$. It follows immediately that \(|A| \leq 2\). Furthermore, for a given \(t\), either \(f_i(0)=t \Rightarrow s_i=0\), or  \(f_i(1)=t \Rightarrow s_i=1\), meaning that there is at most one $s$ with all outputs $t$. Thus we have the following options: 
\begin{enumerate}
    \item $A=\varnothing$ yielding $D=0$;
    \item  $A=\{s\}$ yielding $D=+1$ if $s$ is even and $D=-1$ if $s$ is odd;
    \item $A=\{s^{(t)},s^{(t')}\}$ for distinct realizable trits $t\neq t'$, such that, $s^{(t')}=s^{(t)}\oplus 1^N$ are bitwise complements. Then
    \begin{enumerate}
                \item  For \(N-\)odd, if \(s^{(t')}\) is even/odd then \(s^{(t)}\) is odd/even, yielding \(|E\cap A|=|O\cap A|= 1\) such that \(D=0\);
        \item  For \(N-\)even, if \(s^{(t')}\) is even/odd then \(s^{(t)}\) is even/odd, and , yielding \(|E\cap A|=2, |O\cap A|= 0\) and \(|E\cap A|=0, |O\cap A|= 2\), respectively giving \(D=2\) and \(D=-2\).
    \end{enumerate}
\end{enumerate}

From the previous we can establish that for \(N\)-odd \(D\leq 1\) and for \(N\)-even \(D\leq 2\). We now show how one can achieve these bounds. First for the odd-case, fix any trit e.g. $t=0$ and set $f_i(0)=0$ for all $i$, while choosing the values
$f_i(1)\neq 0$ so they are not all equal (e.g.\ $f_1(1)=\cdots=f_{n-1}(1)=1$, $f_n(1)=2$).
Then $A=\{0^n\}$, which is even, so $D=1$ and therefore $S=2^{n-1}+1$. For the even-case, fix $t_0$ and $t_1$ and set $\{f_i(0)=t_0,f_i(1)=t_1\}$ for all $i$, such that $A=\{0^N, 1^N\}$, which are both even parity strings, since \(n\) is even, thus giving $S=2^{N-1}+2$.

Finally, substituting back to the classical value we get  for \(E(N) := [N\equiv 0 \pmod{2}]\)\[\omega_c(N) = \frac{1}{2}+\frac{2^{E(N)}}{2^N}.\]

Notice also, that although the previous classical bound was obtained by considering explicitly trits as the players' outputs, that is, the scenario in question is \((N,2,3)\), since we are dealing with a unanimous inequality and from Theorem \ref{Th3} we know that the local bound will be constant for \(k \geq m+1\) outputs, then the previous result is valid for \((N,2,k)\) with \( k\geq 3\).

\end{proof}

We will transform the previous local bound for the family (that is, the classical value of the family seen as a unanimous nonlocal game (\(G_{N2}\))  into the local bound corresponding to the form of the family (\(\mathcal{F}_{N2}\)) as it appears it Eq. \eqref{Unanimous Family}. We will do so by using the generic unanimous inequality/game transformation defined in Appendix \ref{section:All unanimous inequalities are nonlocal games}. First, we write all terms $p(=_{ABC..} |x_1 x_2 x_3... ) $ which have negative coefficients in the inequality as $1 -  p(\neq_{ABC..} |x_1 x_2 x_3... ) $, leading to an inequality with only positive coefficients (and a global shift of $-\sum_{\mathbf{x}|\beta_{\mathbf{x}} < 0} \beta_x$, here simply equal to half the number of the terms: $-2^{N-1}$), which correspond to the (relative) values of the prior $\mu(x_1,x_2,..)$ and the predicate  presented previous in Theorem \ref{Family Game Form} (that is, the winning condition on inputs corresponding to positive coefficients is that all parties output the same outcomes, and its negation for inputs corresponding to negative coefficients). Finally, we just have to renormalize the prior, by dividing everything by the sum of all coefficients. 
Therefore, the score of the game on any behavior $p$ is given by
\[
\mathcal{F}_{N2}(p) = 2^N G_{N2}(p)  - 2^{N-1} = 2^N (G_{N2}(p)  - 1/2)
\]

Hence, the local bound of the game given in Theorem \ref{Family Game Form} implies
\begin{equation}
    \begin{split}
        \mathcal{F}_{N2}(p_{\text{local}}) &\leq 2^N \left(1/2 + \frac{2^{E(N)}}{2^N}  - 1/2\right) = 2^{E(N)} \\ &= 1 + (N+1)_{\text{mod} \, 2}
    \end{split}
\end{equation}

\emph{Quantum violations ---}
Let us write explicitly the simplest useful instance of $\mathcal{F}_{N2}$, which is $\mathcal{F}_{32}$:
\begin{equation} \begin{split}
\mathcal{F}_{32}  = &p(=_{ABC}|000 ) - p(=_{ABC}|001 )\\&- p(=_{ABC}|010)  + p(=_{ABC}|011)  \\&- p(=_{ABC}|100 ) + p(=_{ABC}|101 )\\&+  p(=_{ABC}|110) ) - p(=_{ABC}|111 )
     \leq  1
\end{split} 
\end{equation}

Optimizing over quantum strategies for 3-qubit and 3-qutrit states we find: 
\begin{equation}
  Q_2(\mathcal{F}_{32})  \geq 1.0688;\; 
  Q_3(\mathcal{F}_{32}) \geq 1.2568.
\end{equation}

Similarly, we find quantum violations for $N=5$ parties, using $\mathcal{F}_{52}$, which satisfies $\mathcal{F}_{52} \leq 1$:
\begin{equation}
 Q_2(\mathcal{F}_{52})  \geq 1.0016; \; 
  Q_3(\mathcal{F}_{52}) \geq 1.2448.
\end{equation}

\section{No-signaling vertices of the unanimous scenario}
\label{app:nosignalling}

In this section we exhibit the vertices of the no-signaling polytope for arbitrary unanimous scenarios. The reason why such a construction is doable is due to the existence of a \emph{saturating scenario}, similar to the phenomena underlined in Theorem \ref{Th1} and Corollary \ref{cor:bipartite_saturated}. We focus on the notion of "single-party no-signaling", i.e. the weakest form of multipartite no-signaling.

\subsection{Bipartite scenarios }

In the bipartite case we can formulate the following theorem:
\begin{theorem} \label{Th:NS pseudo-telepathy bipartite}
    For any $(2,m,k)$ unanimous game $\mathcal{U}$, one has $NS(\mathcal{U}) = 1$, where $NS(\mathcal{U})$ is the maximum value for game $\mathcal{U}$ allowed by no-signaling strategies. That is, no-signaling resources trivialize unanimous games. 
\end{theorem}
\begin{proof}
    The proof is constructive. We explicitly exhibit a 2-outcome no-signaling strategy achieving a perfect winning probability, for an arbitrary bipartite unanimous game. Let $V_U(a,b,x,y) \in \{0,1\}$ be an arbitrary predicate a of a bipartite unanimous game. That is,

\[
V_U(a,b,x,y) = 
    \begin{cases}
       \delta_{a,b}   &  (x,y) \in S_{=} \\
       1 -   \delta_{a,b}  & (x,y) \not{\in} S_{=} 
     \end{cases}
\]

where $S_{=}$ thus denotes the set of pairs $(x,y)$ for which the parties have to output the same outcomes. Consider the following $2-$outcome behavior:

\[
p_U(a,b|x,y) = \\
    \begin{cases}
       1/2 (\delta_{a=0,b=0} + \delta_{a=1,b=1})   &  (x,y) \in S_{=} \\
       1/2 (\delta_{a=0,b=1} + \delta_{a=1,b=0})   & (x,y) \not{\in} S_{=} 
     \end{cases}
\]

It is easy to see that $p_U(a,b|x,y)$ perfectly satisfies predicate $V_U(a,b,x,y)$, as it outputs only equal outcomes for $(x,y) \in S_{=}$ ($0,0$ with probability $1/2$ and $1,1$ with probability $1/2$), and only different outcomes for $(x,y) \not{\in} S_{=}$ . Therefore, for any valid prior $\mu(x,y)$ associated to predicate $V_U(a,b,x,y)$, behavior $p_U(a,b|x,y)$ achieves a perfect score (i.e. winning probability) of $1$.  Moreover, $p_U(a,b|x,y)$ is no-signaling as it contains only two kinds of probability distributions which are locally the same: 
\begin{align*}
&p_A(a|x) = \sum_b p_U(a,b|x,y) =\sum_b 1/2 (\delta_{a=0,b=0} + \delta_{a=1,b=1}) \\
& = \sum_b 1/2 (\delta_{a=0,b=1} + \delta_{a=1,b=0}) = 1/2; \\
&p_B(b|y) = \sum_a p_U(a,b|x,y) = \sum_a 1/2 (\delta_{a=0,b=0} + \delta_{a=1,b=1})\\
& = \sum_a 1/2 (\delta_{a=0,b=1} + \delta_{a=1,b=0}) = 1/2
\end{align*}
\end{proof}
 This proves that no new vertices can emerge when considering more outcomes. Indeed, one could consider an inequality saturated only for one of these new vertices (any vertex is extremal, therefore there exist hyperplanes touching only one vertex), however, we just saw that $2$ outcomes are always enough to saturate the no-signaling bound for \textit{any} bipartite unanimous inequality. We formalize this in the following Corollary.
 
\begin{corollary}
    \label{Th:NS saturation bipartite}
    For any $(2,m,k)$ unanimous scenario, the no-signaling polytope $\mathcal{NS}_U(2,m,k)$ is saturated for $k=2$ outcomes. That is, $\mathcal{NS}_U$ obtained from $(2,m,k)$ coincides with that from $(2,m,2)$ for all $k \ge 2$.
\end{corollary}

\subsection{Multipartite scenario}

Now generalizing to arbitrary multipartite scenarios we can formulate the following theorem:

\begin{theorem} \label{Th:NS pseudo-telepathy multipartite}
    For any $(n,m,k)$ unanimous game $\mathcal{U}$, one has single-$NS(\mathcal{U}) = 1$, where single-$NS(\mathcal{U})$ is the maximum value allowed by single-party no-signaling strategies.  
\end{theorem}

Where ``single-party no-signaling" means that every party's marginals are independent of all other parties' inputs, that is:

\begin{equation*}
    p(a_i|x_i) = \sum_{a_k \setminus a_i } p(\mathbf{a}|\mathbf{x}) = \sum_{a_k \setminus a_i } p(\mathbf{a}|\mathbf{x'})
\end{equation*}

for all $\mathbf{x},\mathbf{x'}$ with $x_i = x'_i$.

\begin{proof}
  The proof like the one for bipartite is constructive. Let $V(\mathbf{a},\mathbf{x}) \in \{0,1\}$ be an arbitrary predicate a of a unanimous game. That is, 
\begin{equation*}
V(\mathbf{a},\mathbf{x}) = 
    \begin{cases}
       \delta_{a_1 , a_2 , a_3,.. }  &  \mathbf{x} \in S_{=} \\
       1 -  \delta_{a_1 , a_2 , a_3,.. }  & \mathbf{x} \not{\in} S_{=} 
     \end{cases}
\end{equation*}
where $S_{=}$ thus denotes the set of inputs for which the parties all have to output the same outcomes. Now take the following behavior
\begin{equation*}
\begin{aligned}
&p_U(\mathbf{a},\mathbf{x})
=\\
& =\begin{cases}
\frac12 \displaystyle\sum_{a=0}^{1}\delta_{a_1=a,\;a_2=a,\;a_3=a,\;\ldots},
& \mathbf{x}\in S_{=},\\[4pt]
\displaystyle\frac{1}{2^n-2}\Bigl(1-
\sum_{a=0}^{1}\delta_{a_1=a,\;a_2=a,\;a_3=a,\;\ldots}\Bigr),
& \mathbf{x}\notin S_{=}.
\end{cases}
\end{aligned}
\end{equation*}

that is, $p_U(\mathbf{a},\mathbf{x})$ has all outputs equal on $S_{=}$ (all $0$ with probability $1/2$, all $1$ with probability $1/2$), and never have all outputs equal otherwise (with uniform distribution over all remaining possibilities). Therefore, for any valid prior $\mu(\mathbf{x})$ associated to predicate $V(\mathbf{a},\mathbf{x})$, behavior $p_U(\mathbf{a},\mathbf{x}))$ achieves a perfect score (i.e. has a winning probability) of $1$.

Moreover, $p_U(\mathbf{a},\mathbf{x})$ is single-party no-signaling, as all single-party marginals are perfectly random:

\begin{align*}
p_A(a_i|x_i) = \sum_{a_k \setminus a_i } p_U(\mathbf{a}|\mathbf{x}) = 1/2 
\end{align*}
since both probability distributions appearing in $p_U(\mathbf{a}|\mathbf{x})$ are symmetric with respect to outcomes $0$ and $1$. 
  \end{proof}
Thus, just like in the bipartite case  $2$ outcomes are always enough to saturate the single-party no-signaling bound of arbitrary multipartite unanimous inequalities. Again, this allows us to immediately establish the following result:
\begin{corollary}
    \label{Th:NS saturation multipartite}
    For any $(n,m,k)$ unanimous scenario, the single-party no-signaling polytope single-$\mathcal{NS}_U(n,m,k)$ is saturated for $k=2$ outcomes. That is, single-$\mathcal{NS}_U$ obtained from $(n,m,k)$ coincides with that from $(n,m,2)$ for all $k \ge 2$.
\end{corollary}

\end{document}

%% file: Figures/tripartiteeven.tikz
\begin{tikzpicture}[scale=0.8, every node/.style={font=\large}, transform shape]

\tikzset{
  v/.style={circle, fill=black, inner sep=2.6pt},
  ring/.style={circle, dotted, draw=black, line width=0.8pt, inner sep=6pt},
  edge/.style={line width=2pt},
  partlabel/.style={font=\LARGE},
  groupline/.style={draw, line width=1.2pt}
}

\coordinate (A0) at (-1.3,3.0);
\coordinate (A1) at (0.0,3.0);
\coordinate (A2) at (1.3,3.0);

\coordinate (B0) at (-3.0,1.2);
\coordinate (B1) at (-2.2,0.2);
\coordinate (B2) at (-1.4,-0.8);

\coordinate (C0) at ( 1.4,-0.8);
\coordinate (C1) at ( 2.2, 0.2);
\coordinate (C2) at ( 3.0, 1.2);

\def\GroupScale{0.8} 
\def\GroupXRadius{2.35*\GroupScale}
\def\GroupYRadius{0.8*\GroupScale}
\draw[groupline] (A1) ellipse [x radius=\GroupXRadius, y radius=\GroupYRadius];

\begin{scope}[rotate around={-51.34:(B1)}]
  \draw[groupline] (B1) ellipse [x radius=\GroupXRadius, y radius=\GroupYRadius];
\end{scope}

\begin{scope}[rotate around={ 51.34:(C1)}]
  \draw[groupline] (C1) ellipse [x radius=\GroupXRadius, y radius=\GroupYRadius];
\end{scope}

\node[v] (a0) at (A0) {};
\node[v] (a1) at (A1) {};
\node[v] (a2) at (A2) {};

\node[v] (b0) at (B0) {};
\node[v] (b1) at (B1) {};
\node[v] (b2) at (B2) {};

\node[v] (c0) at (C0) {};
\node[v] (c1) at (C1) {};
\node[v] (c2) at (C2) {};

\node (A0lbl) at ($(A0)+(0,0.95)$) {0};
\node (A1lbl) at ($(A1)+(0,0.95)$) {1};
\node (A2lbl) at ($(A2)+(0,0.95)$) {2};

\node (B0lbl) at ($(B0)+(-1.0,0.0)$) {0};
\node (B1lbl) at ($(B1)+(-1.0,-0.6)$) {1};
\node (B2lbl) at ($(B2)+(-1.0,-0.6)$) {2};

\node (C0lbl) at ($(C0)+(1.0,0.0)+(-0.25,-0.35)$) {0};
\node (C1lbl) at ($(C1)+(1.0,-0.6)+(-0.25,-0.35)$) {1};
\node (C2lbl) at ($(C2)+(1.0,-0.6)+(-0.25,-0.35)$) {2};

\node[ring] at (A2) {};
\node[ring] at (B2) {};
\node[ring] at (C2) {};

\node[partlabel] at ($(A1)+(0,1.55)$) {A}; \node[partlabel] at ($(B1)+(-1.55,-0.95)$) {B}; \node[partlabel] at ($(C1)+( 1.80,-0.95)$) {C};
\draw[edge] (a0) -- (b0);
\draw[edge] (a1) -- (c1);
\draw[edge] (b1) -- (c0);

\end{tikzpicture}

%% file: Figures/tripartiteodd.tikz
\begin{tikzpicture}[scale=0.7, every node/.style={font=\large}, transform shape]

\tikzset{
  v/.style={circle, fill=black, inner sep=2.6pt},
  ring/.style={circle, dotted, draw=black, line width=0.8pt, inner sep=6pt},
  dashring/.style={circle, dashed, draw=black, line width=0.8pt, inner sep=6pt},
  edge/.style={line width=2pt},
  partlabel/.style={font=\LARGE},
  groupline/.style={draw, line width=1.2pt}
}

\coordinate (A0) at (-1.3,3.0);
\coordinate (A1) at (0.0,3.0);
\coordinate (A2) at (1.3,3.0);
\coordinate (A3) at (2.6,3.0);

\coordinate (B0) at (-3.0,1.2);
\coordinate (B1) at (-2.2,0.2);
\coordinate (B2) at (-1.4,-0.8);
\coordinate (B3) at (-0.6,-1.8);

\def\Cdown{0,-0.35}

\coordinate (C0) at ($(1.4,-0.8)+(\Cdown)$);
\coordinate (C1) at ($(2.2, 0.2)+(\Cdown)$);
\coordinate (C2) at ($(3.0, 1.2)+(\Cdown)$);
\coordinate (C3) at ($(3.8, 2.2)+(\Cdown)$);

\coordinate (Acent) at ($(A0)!0.5!(A3)$);
\coordinate (Bcent) at ($(B0)!0.5!(B3)$);
\coordinate (Ccent) at ($(C0)!0.5!(C3)$);

\def\GroupScale{0.8}
\def\GroupXRadius{3.55*\GroupScale}
\def\GroupYRadius{0.95*\GroupScale}

\draw[groupline] (Acent) ellipse [x radius=\GroupXRadius, y radius=\GroupYRadius];

\begin{scope}[rotate around={-51.34:(Bcent)}]
  \draw[groupline] (Bcent) ellipse [x radius=\GroupXRadius, y radius=\GroupYRadius];
\end{scope}

\begin{scope}[rotate around={ 51.34:(Ccent)}]
  \draw[groupline] (Ccent) ellipse [x radius=\GroupXRadius, y radius=\GroupYRadius];
\end{scope}

\node[v] (a0) at (A0) {};
\node[v] (a1) at (A1) {};
\node[v] (a2) at (A2) {};
\node[v] (a3) at (A3) {};

\node[v] (b0) at (B0) {};
\node[v] (b1) at (B1) {};
\node[v] (b2) at (B2) {};
\node[v] (b3) at (B3) {};

\node[v] (c0) at (C0) {};
\node[v] (c1) at (C1) {};
\node[v] (c2) at (C2) {};
\node[v] (c3) at (C3) {};

\def\Aup{1.10} 
\node[anchor=south] (A0lbl) at ($(A0)+(0,\Aup)$) {0};
\node[anchor=south] (A1lbl) at ($(A1)+(0,\Aup)$) {1};
\node[anchor=south] (A2lbl) at ($(A2)+(0,\Aup)$) {2};
\node[anchor=south] (A3lbl) at ($(A3)+(0,\Aup)$) {3};

\node (B0lbl) at ($(B0)+(-1.0,0.0)$) {0};
\node (B1lbl) at ($(B1)+(-1.0,-0.6)$) {1};
\node (B2lbl) at ($(B2)+(-1.0,-0.6)$) {2};
\node (B3lbl) at ($(B3)+(-1.0,-0.6)$) {3};

\def\Cshift{-0.25,-0.35}
\node[anchor=west] (C0lbl) at ($(C0)+(1.0,-0.6)+(\Cshift)$) {0};
\node[anchor=west] (C1lbl) at ($(C1)+(1.0,-0.6)+(\Cshift)$) {1};
\node[anchor=west] (C2lbl) at ($(C2)+(1.0,-0.6)+(\Cshift)$) {2};
\node[anchor=west] (C3lbl) at ($(C3)+(1.0,-0.6)+(\Cshift)$) {3};

\node[ring] at (A3) {};
\node[ring] at (B3) {};
\node[ring] at (C3) {};

\node[dashring] at (C2) {};

\node[partlabel] at ($(Acent)+(0,2)$) {A};
\node[partlabel] at ($(Bcent)+(-1.55,-1.15)$) {B};
\node[partlabel] at ($(Ccent)+( 1.55,-1.35)$) {C};

\draw[edge] (a0) -- (b0);
\draw[edge] (a1) -- (c1);
\draw[edge] (b1) -- (c0);
\draw[edge] (a2) -- (b2);

\end{tikzpicture}